\documentclass[a4paper,11pt,oneside,reqno]{amsart}
\usepackage[utf8]{inputenc}
\usepackage{amsmath,amsthm,amsfonts,latexsym,amssymb,bm,enumerate}
\usepackage{ae}
\usepackage{cite}
\usepackage{float}
\usepackage{lmodern}
\usepackage[T1]{fontenc}

\usepackage{color}

\usepackage[colorlinks=true]{hyperref}

\definecolor{dark-red}{rgb}{.54,.0,.0}
\definecolor{dark-green}{rgb}{.0,.4,.0}
\definecolor{dark-blue}{rgb}{.04,.04,.4}

\hypersetup{linkcolor=dark-red, urlcolor=dark-blue, citecolor=dark-green}

\usepackage[dvips]{graphicx}
\usepackage{psfrag}
\DeclareGraphicsExtensions{.eps,.art,.ART,.ps}

\usepackage{rotating}

\newcounter{mnotecount}[section]

\usepackage[dvips]{graphicx}
\usepackage{psfrag}
\DeclareGraphicsExtensions{.eps,.art,.ART,.ps}

\DeclareFontFamily{OT1}{pzc}{}
\DeclareFontShape{OT1}{pzc}{m}{it}%
              {<-> s * pzcmi8t}{}
\DeclareMathAlphabet{\mathpzc}{OT1}{pzc}%
                                {m}{it}

\newtheorem{Thm}{Theorem}[section]

\newtheorem{Lem}[Thm]{Lemma}
\newtheorem{Cor}[Thm]{Corollary}
\newtheorem{Prop}[Thm]{Proposition}

\newtheorem{Rmk}[Thm]{Remark}

\newcommand{\cal}{\mathcal}

\renewcommand{\Psi}{\rho}

\newcommand{\R}{\mathbb{R}}

\newcommand{\eps}{\varepsilon}

\newcommand{\vlinha}{\frac{e^2}{r}+\frac{\Lambda}{3}r^3-\omega}
\newcommand{\vlinhazz}{\frac{e^2}{r}+\frac{\Lambda}{3}r^3-\omega_0}
\newcommand{\vlinhaz}{\frac{e^2}{{r}_0(0)}+\frac{\Lambda}{3}{r}_0^3(0)-\omega_0}
\newcommand{\muu}{1-\frac{2\omega}{r}+\frac{e^2}{r^2}-\frac{\Lambda}{3}r^2}
\newcommand{\muuz}{1-\frac{2\omega_0}{\tilde{r}_0}+\frac{e^2}{\tilde{r}_0^2}-\frac{\Lambda}{3}\tilde{r}_0^2}
\newcommand{\truemu}{\frac{2\omega}{r}-\frac{e^2}{r^2}+\frac{\Lambda}{3}r^2}
\newcommand{\mysigma}{{\cal P}}
\renewcommand{\omega}{\varpi}
\newcommand{\myumax}{U}
\newcommand{\myU}{U'}

\newcommand{\ru}{\partial_r(1-\mu)(r_+,\varpi_0)}
\newcommand{\ckrm}{\check r_-}
\newcommand{\ckrp}{\check r_+}
\newcommand{\hatz}{\widehat{\frac{\zeta}{\nu}}}

\newcommand{\cg}{\Gamma}
\newcommand{\vgr}{v_{\ckrm}}
\newcommand{\ugr}{u_{\ckrm}}
\newcommand{\ugrp}{u_{\ckrp}}
\newcommand{\vgrp}{v_{\ckrp}}
\newcommand{\ug}{u_r}
\newcommand{\vg}{v_r}
\newcommand{\ugs}{u_s}
\newcommand{\vgs}{v_s}
\newcommand{\vgsd}{v_{s_2}}
\newcommand{\ugz}{u_{r_0}}
\newcommand{\vgz}{v_{r_0}}
\newcommand{\ckr}{\check{r}}
\newcommand{\uckr}{u_{\ckr}}
\newcommand{\vckr}{v_{\ckr}}
\newcommand{\udz}{u_{r_+-\delta}}
\newcommand{\vdz}{v_{r_+-\delta}}
\newcommand{\udp}{u'_{r_+-\delta}}
\newcommand{\gam}{\gamma}
\newcommand{\ugam}{u_{\gam}}
\newcommand{\vgam}{v_{\gam}}

\begin{document}

\newcounter{enumii_saved}

\title[Global uniqueness with a cosmological constant - Part 2]{On the global uniqueness for the Einstein-Maxwell-scalar field system with a cosmological constant \\ \vspace{.2cm}
	{\small Part 2. Structure of the solutions and stability of the Cauchy horizon}}

\author{Jo\~ao L.~Costa}
\author{Pedro M.~Gir\~ao}
\author{Jos\'{e} Nat\'{a}rio}
\author{Jorge Drumond Silva}

\address{Jo\~ao L.~Costa: 
ISCTE - Instituto Universitário de Lisboa, Portugal
and 
Center for Mathematical Analysis, Geometry and Dynamical Systems,
Instituto Superior T\'ecnico, Universidade de Lisboa, Portugal
}
\email{jlca@iscte.pt}

\address{Pedro M.~Gir\~ao, Jos\'{e} Nat\'{a}rio and Jorge Drumond Silva: 
Center for Mathematical Analysis, Geometry and Dynamical Systems,
Instituto Superior T\'ecnico, Universidade de Lisboa, Portugal}
\email{pgirao@math.ist.utl.pt}
\email{jnatar@math.ist.utl.pt}
\email{jsilva@math.ist.utl.pt}

\thanks{Partially funded by FCT/Portugal through project PEst-OE/EEI/LA0009/2013.
P.~Girão and J.~Silva were also partially funded by FCT/Portugal through grants PTDC/MAT114397/2009 and UTA$\underline{\ }$CMU/MAT/0007/2009.}

\maketitle

\begin{center}
{\bf Abstract}
\end{center}

This paper is the second part of a trilogy dedicated to the following problem: 
given spherically symmetric characteristic
initial data for the Einstein-Maxwell-scalar
field system with a cosmological constant $\Lambda$,
with the data on the outgoing initial null
hypersurface given by a subextremal Reissner-Nordstr\"{o}m black
hole event horizon, study the future extendibility of the corresponding
maximal globally hyperbolic development as a  ``suitably regular'' Lorentzian manifold. 

In the first paper of this sequence \cite{relIst1}, we established well posedness of the characteristic problem with general initial data.

In this second paper, we generalize the results of Dafermos~\cite{Dafermos1} on the stability of the radius function at the Cauchy horizon by including a cosmological constant.
This requires a considerable deviation from the strategy followed in~\cite{Dafermos1}, focusing on the level sets of the radius function instead of the red-shift and blue-shift regions.
We also present new results on the global structure of the solution when the free data is not identically zero in a neighborhood of the origin. 

In the third and final paper \cite{relIst3}, we will consider the issue of mass inflation and extendibility of solutions beyond the Cauchy horizon.


\newpage
{
\setcounter{tocdepth}{1}
\tableofcontents
}
\section{Introduction}

This paper is the second part of a trilogy dedicated to the following problem: 
given spherically symmetric characteristic
initial data for the Einstein-Maxwell-scalar
field system with a cosmological constant $\Lambda$,
with the data on the outgoing initial null
hypersurface given by a subextremal Reissner-Nordstr\"{o}m black
hole event horizon, 
and the remaining data otherwise free,
study the future extendibility of the corresponding
maximal globally hyperbolic development as a  ``suitably regular'' Lorentzian manifold.
We are motivated by the strong cosmic censorship conjecture and the question of determinism in general relativity.
As explained in detail in the Introduction of Part~1, strong cosmic censorship is one of the most fundamental open problems in
general relativity (see the classic monographs \cite{ChruscielSCC, Earman} and the
discussions in \cite{ChristodoulouGlobalnew, Dafermos1, DafermosBlack} for the general context of this problem).
Although significant developments have been achieved in the last five
decades (from the initial heuristic works \cite{SimpsonInternal, IsraelPoisson} to rigorous mathematical results \cite{Dafermos1, Dafermos2, DafermosProof}), including some recent encouraging progress (see \cite{Franzen, LukInst, DafermosBlack} and references therein), a complete resolution of the conjecture at hand still seems out
of reach. Nonetheless, the spherically symmetric self-gravitating scalar field model has provided considerable insight into the harder problem of vacuum collapse without symmetries \cite{Christodoulou:2008}; this was explored in \cite{LukWeak} to obtain the first promising steps towards understanding the stability of Cauchy horizons without symmetry assumptions.

In Part~1, we established the equivalence (under appropriate regularity conditions for the initial data) between the Einstein equations~\eqref{wave_r}$-$\eqref{wave_Omega}
and the system of first order PDE~\eqref{r_u}$-$\eqref{kappa_at_u}. We proved existence, uniqueness and identified a breakdown criterion for solutions of this system (see Section~\ref{main-one}).

In the current paper we are concerned with the structure of the solutions of the characteristic problem, 
and wish to
address the question of existence and stability of the Cauchy horizon
when the initial data is as above. This is intimately related to the issue of global uniqueness for the Einstein equations: it is the possibility of extension of 
solutions across this horizon that leads to the breakdown of global uniqueness and, in case the phenomenon persists for generic initial data, to the failure of the strong cosmic 
censorship conjecture. 

As in~\cite{Dafermos1}, we introduce a certain generic element in the formulation
 of our
problem by
perturbing a subextremal Reissner-Nordstr\"{o}m black
hole (whose Cauchy horizon formation is archetypal) by arbitrary characteristic data along the ingoing null direction. The study of the conditions under 
which the solutions can be extended across the Cauchy horizon is left to Part~3.

We take many ideas from~\cite{Dafermos1} and~\cite{Dafermos2} and build on these works.
In particular, we borrow the following three very important techniques. 
(i)~The partition of the spacetime domain of the solution into four regions and the construction of a carefully chosen spacelike curve to separate the last two.
(ii)~The use of the Raychaudhuri equation in~$v$ to estimate $\frac{\nu}{1-\mu}$ at a larger~$v$ from its value at a smaller~$v$.
(iii)~The use of BV estimates for the field.

Nonetheless, the introduction of a cosmological constant $\Lambda$ causes a significant
difference that requires deviation 
from the original strategies developed in~\cite{Dafermos1} and~\cite{Dafermos2}.
Moreover, we introduce some technical simplifications and obtain sharper and more detailed estimates.
These improvements will be crucial for our arguments in Part~3. 

Our approach therefore has three main departures from
the one of Dafermos:
\begin{enumerate}[i)]
	\item First, due to the presence of the cosmological constant $\Lambda$, the curves of constant shift, which are used in~\cite{Dafermos1} and~\cite{Dafermos2}, 
        are no longer necessarily spacelike for $\Lambda>0$ large. This forces us to find an alternative approach;
	we have chosen to work with curves of constant $r$ coordinate instead of working with curves of constant shift, which turns out to be a simpler approach.
	Furthermore, it allows us to treat the cases $\Lambda<0$, $\Lambda=0$ and $\Lambda>0$ in a unified framework.
	\item Second, we show that the Bondi coordinates $(r,v)$ are the ones most adapted to estimating the growth of the fields
	as we progress away from the event horizon.
	Our approach starts by controlling the field $\frac\zeta\nu$ using~\eqref{triple}. Although this is similar to~\eqref{double}, there is one distinction
	which makes all the difference. It consists of the fact that in the double integral in~\eqref{double} the field $\frac\zeta\nu$ is
	multiplied by the function $\nu$. When we pass to Bondi coordinates this function disappears, making a simple application of Gronwall's inequality,
	such as the one we present, possible. This would not work in the double null coordinate system~$(u,v)$.
	\item Third, our estimates are not subordinate to the division of the solution spacetime into red shift, no shift and blue shift regions. Instead, we consider the regions $\{r \geq \ckrp\}$, $\{\ckrm \leq r \leq \ckrp\}$ and $\{r \leq \ckrm\}$, where $\ckrp$ is smaller than but sufficiently close to the radius $r_+$ of the Reissner-Nordstr\"om event horizon, and $\ckrm$ is bigger than but sufficiently close to the radius $r_-$ of the Reissner-Nordstr\"om Cauchy horizon. These may be loosely thought of as red shift, no shift and blue shift regions of the background Reissner-Nordstr\"om solution, even though the shift factor is not small and indeed changes significantly from red to blue in the intermediate region.
\end{enumerate}

Our first objective is to obtain good upper bounds for $-\lambda$ in the different regions of spacetime. These will enable us to show that
the radius function $r$ is bounded below by a positive constant. However, good estimates for $-\nu$ and the fields $\theta$ and $\zeta$ will also be essential in Part~3.

The main result of this paper is therefore
\begin{Thm}\label{r-stability}
Consider the characteristic initial value problem for the first order system of PDE~\eqref{r_u}$-$\eqref{kappa_at_u} 
 with initial data\/~\eqref{iu}$-$\eqref{iv} (so that $\{0\}\times[0,\infty[$ is the event horizon of a subextremal Reissner-Nordstr\"{o}m solution with mass $M>0$).
 Assume that $\zeta_0$ is continuous and $\zeta_0(0)=0$. Then there exists $U>0$ such that
the domain $\mysigma$ of the (future) maximal development contains $[0,\myumax]\times[0,\infty[$. Moreover, 
$$
\inf_{[0,\myumax]\times[0,\infty[}r>0,
$$
the limit
$$
r(u,\infty):=\lim_{v \to \infty}r(u,v)
$$
exists for all $u \in \left] 0,U \right]$ and
$$
\lim_{u\searrow 0}r(u,\infty)=r_-.
$$
\end{Thm}
So, under the hypotheses of Theorem~\ref{r-stability},
the argument in~\cite[Section~11]{Dafermos2}, shows that, as in the case when $\Lambda=0$, the spacetime is extendible across the Cauchy horizon with a $C^0$ metric.

We also prove that only in the case of the Reissner-Nordstr\"{o}m solution does the curve $\{ r= r_- \}$ coincide with the Cauchy horizon. As soon as the initial data field is not identically zero, the curve
$\{ r= r_- \}$ is contained in~${\cal P}$ (Theorem~\ref{rmenos}). This is an interesting geometrical condition and it is conceptually relevant given the importance
that we confer to the curves of constant $r$. We also prove that, in contrast with what happens with the Reissner-Nordstr\"{o}m solution,
the presence of any nonzero field immediately causes the integral $\int_0^\infty \kappa(u,v)\,dv$ to be finite for any $u>0$ (Lemma~\ref{l-kappa}).
As a consequence, the affine parameter of any outgoing null geodesic inside the event horizon is finite at the Cauchy horizon (Corollary~\ref{affine}).

\section{Framework and some results from Part~1}\label{main-one}

\subsection*{The spherically symmetric Einstein-Maxwell-scalar field system with a cosmological constant}

Consider a spherically symmetric spacetime with metric
$$g=-\Omega^2(u,v)\,dudv+r^2(u,v)\,\sigma_{\mathbb{S}^2},$$
where $\sigma_{\mathbb{S}^2}$ is the round metric on the 2-sphere.
The Einstein-Maxwell-scalar field system with a cosmological constant $\Lambda$ and total electric charge $4\pi e$ reduces to the following system of equations:
the wave equation for $r$,
\begin{equation}\label{wave_r} 
\partial_u\partial_vr=\frac{\Omega^2}{2}\frac{1}{r^2}\left(\vlinha\right),
\end{equation}
the wave equation for $\phi$,
\begin{equation}\label{wave_phi} 
\partial_u\partial_v\phi=-\,\frac{\partial_ur\,\partial_v\phi+\partial_vr\,\partial_u\phi}{r},
\end{equation}
the Raychaudhuri equation in the $u$ direction,
\begin{equation}\label{r_uu} 
\partial_u\left(\frac{\partial_ur}{\Omega^2}\right)=-r\frac{(\partial_u\phi)^2}{\Omega^2},
\end{equation}
the Raychaudhuri equation in the $v$ direction,
\begin{equation}\label{r_vv} 
\partial_v\left(\frac{\partial_vr}{\Omega^2}\right)=-r\frac{(\partial_v\phi)^2}{\Omega^2},
\end{equation}
and the wave equation for $\ln\Omega$,
 \begin{equation}\label{wave_Omega} 
\partial_v\partial_u\ln\Omega=-\partial_u\phi\,\partial_v\phi-\,\frac{\Omega^2e^2}{2r^4}+\frac{\Omega^2}{4r^2}+\frac{\partial_ur\,\partial_vr}{r^2}.
\end{equation}

\subsection*{The first order system}

Given $r$, $\phi$ and $\Omega$, solutions of the Einstein equations,
let
\begin{equation}\label{nu_0}
\nu:=\partial_u r
\end{equation}
\begin{equation}\label{lambda_0}
\lambda:=\partial_v r,
\end{equation}
\begin{equation}\label{bar_rafaeli} 
\omega:=\frac{e^2}{2r}+\frac{r}{2}-\frac{\Lambda}{6}r^3+\frac{2r}{\Omega^2}\nu\lambda,
\end{equation}
\begin{equation}\label{mu_0} 
\mu:=\truemu,
\end{equation}
\begin{equation}\label{theta} 
\theta:=r\partial_v\phi,
\end{equation}
\begin{equation}\label{zeta} 
\zeta:=r\partial_u\phi
\end{equation}
and
\begin{equation}\label{kappa_0} 
 \kappa:=\frac{\lambda}{1-\mu}.
\end{equation}
Notice that we may rewrite~\eqref{bar_rafaeli} as
\begin{equation}\label{omega_sq}
\Omega^2=-\,\frac{4\nu\lambda}{1-\mu}=-4\nu\kappa.
\end{equation}
The Einstein equations imply the first order system for $(r,\nu,\lambda,\varpi,\theta,\zeta,\kappa)$
\begin{eqnarray} 
 \partial_ur&=&\nu\label{r_u},\\
 \partial_vr&=&\lambda\label{r_v},\\
 \partial_u\lambda&=&\nu\kappa\partial_r(1-\mu)\label{lambda_u},\\
 \partial_v\nu&=&\nu\kappa\partial_r(1-\mu),\label{nu_v}\\
 \partial_u\omega&=&\frac 12(1-\mu)\left(\frac\zeta\nu\right)^2\nu,\label{omega_u}\\
 \partial_v\omega&=&\frac 12\frac{\theta^2}{\kappa},\label{omega_v}\\
 \partial_u\theta&=&-\,\frac{\zeta\lambda}{r},\label{theta_u}\\
 \partial_v\zeta&=&-\,\frac{\theta\nu}{r},\label{zeta_v}\\
 \partial_u\kappa&=&\kappa\nu\frac 1r\left(\frac{\zeta}{\nu}\right)^2,\label{kappa_u}
\end{eqnarray}
with the restriction
\begin{equation}\label{kappa_at_u} 
\lambda=\kappa(1-\mu).
\end{equation}
Under appropriate regularity conditions for the initial data, the system of first order PDE~\eqref{r_u}$-$\eqref{kappa_at_u} also implies the Einstein equations \eqref{wave_r}$-$\eqref{wave_Omega}.

\subsection*{Initial data}
In Part~1 we study well posedness of the first order system for general initial data. In this paper
we take the initial data on the outgoing null direction $v$ to be the data on the event horizon of a subextremal Reissner-Nordstr\"{o}m solution
with mass $M$. The initial data on the ingoing null direction $u$ is free.
More precisely, we choose
\begin{equation}\label{iu}
\left\{
\begin{array}{lclcl}
 r(u,0)&=&r_0(u)&=&r_+-u,\\
 \nu(u,0)&=&\nu_0(u)&=&-1,\\
 \zeta(u,0)&=&\zeta_0(u),&&
\end{array}
\right.\qquad{\rm for}\ u\in[0,U],
\end{equation}
\begin{equation}\label{iv}
\left\{
\begin{array}{lclcl}
 \lambda(0,v)&=&\lambda_0(v)&=&0,\\
 \omega(0,v)&=&\omega_0(v)&=&M,\\
 \theta(0,v)&=&\theta_0(v)&=&0,\\
 \kappa(0,v)&=&\kappa_0(v)&=&1,
\end{array}
\right.\qquad{\rm for}\ v\in[0,\infty[.
\end{equation}
Here $r_+>0$ is the radius of the event horizon. We assume $\zeta_0$ is continuous and $\zeta_0(0)=0$.

\subsection*{Well posedness of the first order system}

Theorem 4.4 of Part~1, for the initial data above, reads:

\begin{Thm}\label{maximal} The characteristic initial value problem\/~\eqref{r_u}$-$\eqref{kappa_at_u}, with initial 
conditions\/~\eqref{iu} and\/~\eqref{iv}, where $\zeta_0$ is continuous and $\zeta_0(0)=0$,
has a unique solution defined on a maximal past set $\mysigma$ containing a neighborhood of $[0,\myumax]\times\{0\}\cup\{0\}\times[0,\infty[$.\end{Thm}

\begin{Rmk}
Notice that the initial data\/~\eqref{iu} and\/~\eqref{iv} satisfies the regularity condition $\rm(h4)$ in\/ {\rm Part~1} (that is, $\nu_0$, $\lambda_0$ and $\kappa_0$ are $C^1$). Therefore the solution of the characteristic initial value problem\/~\eqref{r_u}$-$\eqref{kappa_at_u} corresponds to a classical solution of the Einstein equations \eqref{wave_r}$-$\eqref{wave_Omega}.
\end{Rmk}

\subsection*{Breakdown criterion}

Theorem 5.4 of Part~1, for the initial data above, reads:

\begin{Thm}\label{bdown}
Suppose that $(r,\nu,\lambda,\omega,\theta,\zeta,\kappa)$ is the maximal solution of the characteristic initial value problem\/~\eqref{r_u}$-$\eqref{kappa_at_u}, with initial 
conditions\/~\eqref{iu} and\/~\eqref{iv}.
If $(\myU,V')$ is a point on the boundary of $\mysigma$ with $0<\myU<\myumax$ and $V'>0$, then
for all sequences $(u_n,v_n)$ in $\mysigma$ converging to $(\myU,V')$, we have
$$
r(u_n,v_n)\to 0\quad {\rm and}\quad \omega(u_n,v_n)\to\infty.
$$
\end{Thm}

\subsection*{Reissner-Nordstr\"om solution}

For comparison purposes, we notice that the Reissner-Nordstr\"om solution (with a cosmological constant), obtained from the initial data $\zeta_0(u)=0$, corresponds to
\begin{align}
& \lambda = 1 - \mu, \label{primeira}\\
& \nu = - \, \frac{1 - \mu}{(1 - \mu)(\,\cdot\,,0)}, \label{segunda}\\
& \varpi = \varpi_0, \\
& \kappa = 1, \\ 
& \zeta = \theta = 0.\label{ultima}
\end{align}

\section{Preliminaries on the analysis of the solution}\label{cans}

We now take the initial data on the $v$ axis to be the data on the event horizon of a subextremal Reissner-Nordstr\"{o}m 
solution with mass $M>0$. So, we
choose initial data as in~\eqref{iu}$-$\eqref{iv} with $\zeta_0(0)=0$.
Moreover, we assume $\zeta_0$ to be continuous.
Since in this case the function $\varpi_0$ is constant equal to $M$, we also denote $M$ by $\varpi_0$.
In particular, when $\Lambda<0$, which corresponds to the Reissner-Nordstr\"{o}m anti-de Sitter solution, and when $\Lambda=0$, which corresponds to the Reissner-Nordstr\"{o}m solution, we assume that
$$ 
r\mapsto (1-\mu)(r,\varpi_0)=1-\frac{2\varpi_0}{r}+\frac{e^2}{r^2}-\frac\Lambda 3r^2
$$ 
has two zeros $r_-(\varpi_0)=r_-<r_+=r_+(\varpi_0)$. When $\Lambda>0$, which corresponds to the Reissner-Nordstr\"{o}m de Sitter solution, we assume that
$r\mapsto (1-\mu)(r,\varpi_0)$ has three zeros $r_-(\varpi_0)=r_-<r_+=r_+(\varpi_0)<r_c=r_c(\varpi_0)$.

\begin{center}
\begin{psfrags}
\psfrag{a}{{\tiny $\Lambda=0$}}
\psfrag{b}{{\tiny $\Lambda<0$}}
\psfrag{d}{{\tiny $\Lambda>0$}}
\psfrag{z}{{\tiny $r_0$}}
\psfrag{y}{{\tiny $\ $}}
\psfrag{r}{{\tiny $r$}}
\psfrag{n}{{\tiny \!\!\!\!\!\!\!\!\!\!\!\!\!\!\!\!\!$1-\mu(r,\varpi_0)$}}
\psfrag{p}{{\tiny $r_+$}}
\psfrag{m}{{\tiny \!\!\!\!\!\!$r_-$}}
\psfrag{c}{{\tiny $r_c$}}
\includegraphics[scale=.6]{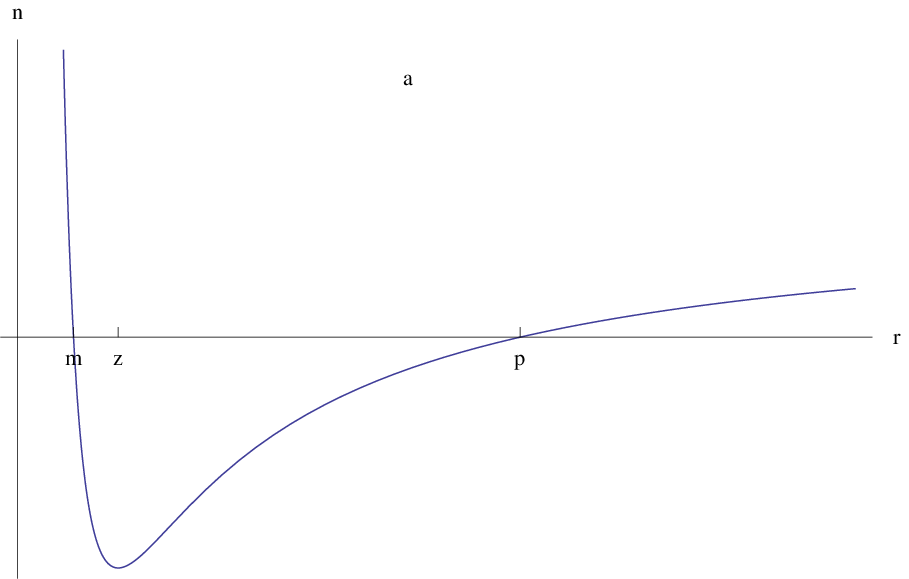}\\
\includegraphics[scale=.6]{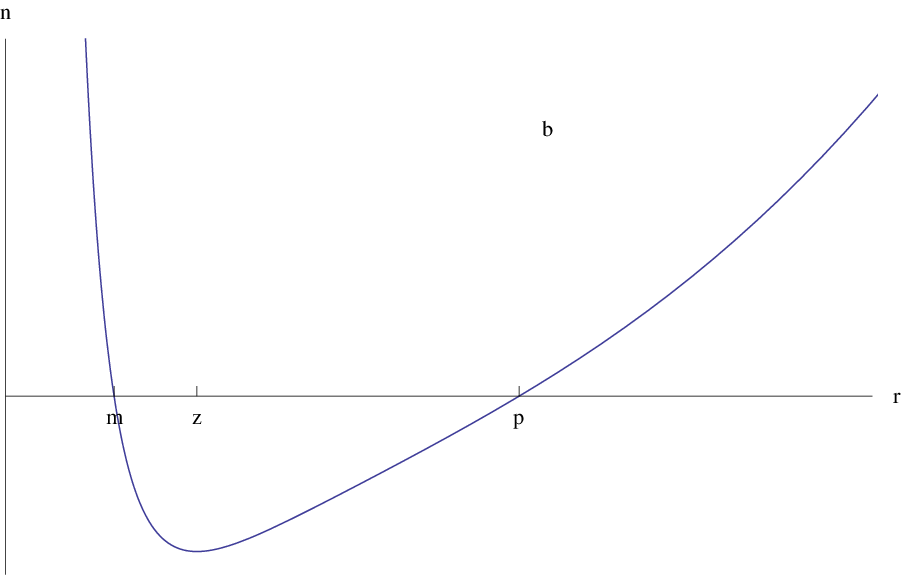}\ \ \ \ 
\includegraphics[scale=.6]{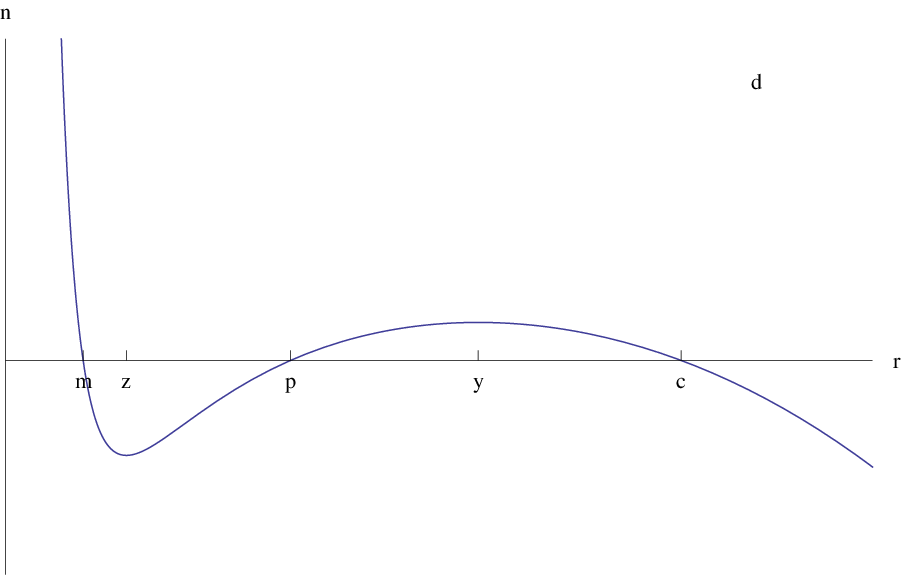}
\end{psfrags}
\end{center}

We define $\eta$ to be the function
$$ 
\eta=\vlinha.
$$ 
The functions $(r,\varpi)\mapsto\eta(r,\varpi)$ and $(r,\varpi)\mapsto (1-\mu)(r,\varpi)$ 
are related by
\begin{equation}\label{eta} 
\eta=-\,\frac{r^2}{2}\partial_r(1-\mu).
\end{equation}
We define the function $\eta_0:\R^+\to\R$ by
$$
\eta_0(r)=\vlinhazz.
$$
We will repeatedly use the fact that $\eta(r,\varpi)\leq\eta_0(r)$ (see Lemma~\ref{sign}).
If $\Lambda\leq 0$, then $\eta_0'<0$. So $\eta_0$ is strictly decreasing and has precisely one zero.
The zero is located between $r_-$ and $r_+$. If $\Lambda>0$, then $\eta_0''$ is positive, so $\eta_0$ is strictly convex
and has precisely two zeros: one zero is located between $r_-$ and $r_+$ and the other zero is located between $r_+$ and $r_c$.
We denote by $r_0$ the zero of $\eta_0$ between $r_-$ and $r_+$ in both cases.

\begin{center}
\begin{psfrags}
\psfrag{a}{{\tiny $\Lambda=0$}}
\psfrag{b}{{\tiny $\Lambda<0$}}
\psfrag{d}{{\tiny $\Lambda>0$}}
\psfrag{e}{{\tiny \!$\eta_0$}}
\psfrag{z}{{\tiny \!\!\!\!\!$r_0$}}
\psfrag{y}{{\tiny $\ $}}
\psfrag{r}{{\tiny $r$}}
\psfrag{p}{{\tiny $r_+$}}
\psfrag{m}{{\tiny \!\!\!\!\!\!$r_-$}}
\psfrag{c}{{\tiny $r_c$}}
\includegraphics[scale=.6]{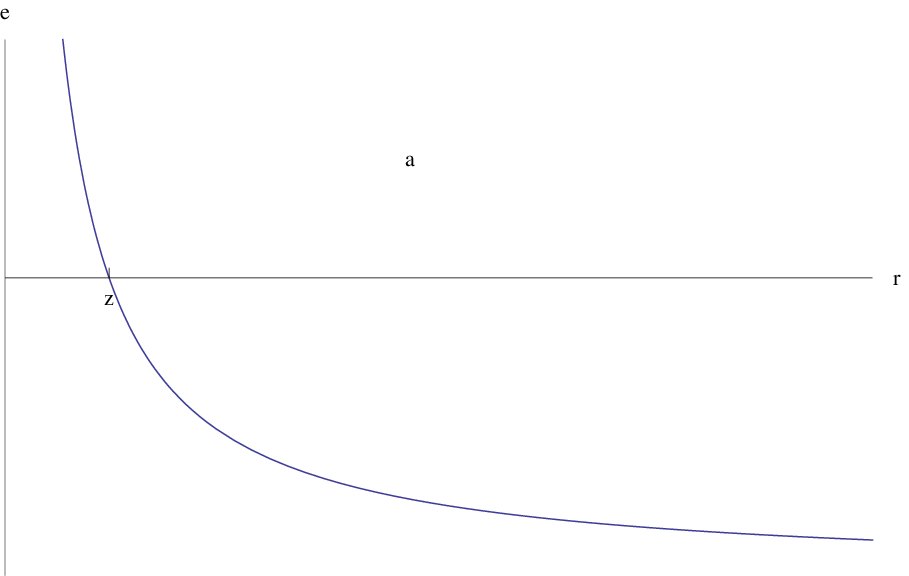}\\
\includegraphics[scale=.6]{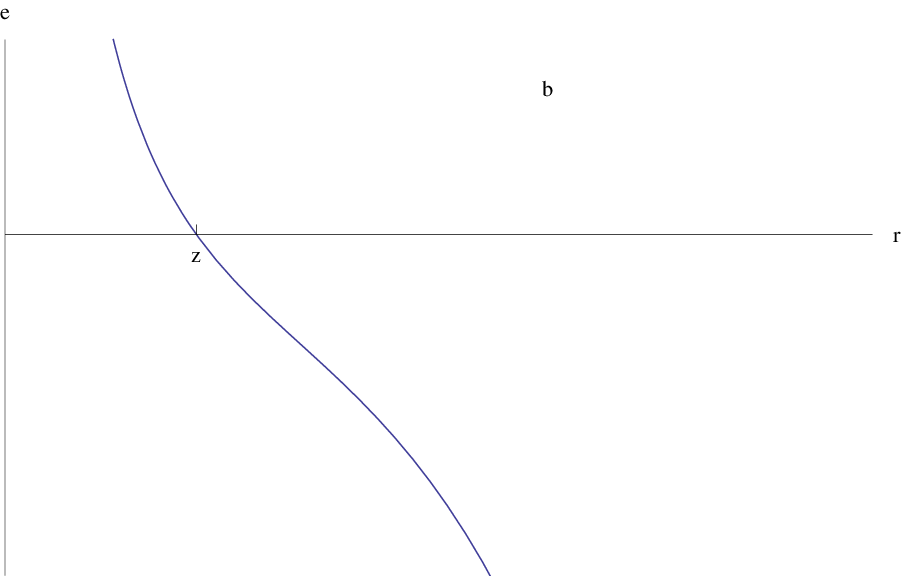}\ \ \ \ 
\includegraphics[scale=.6]{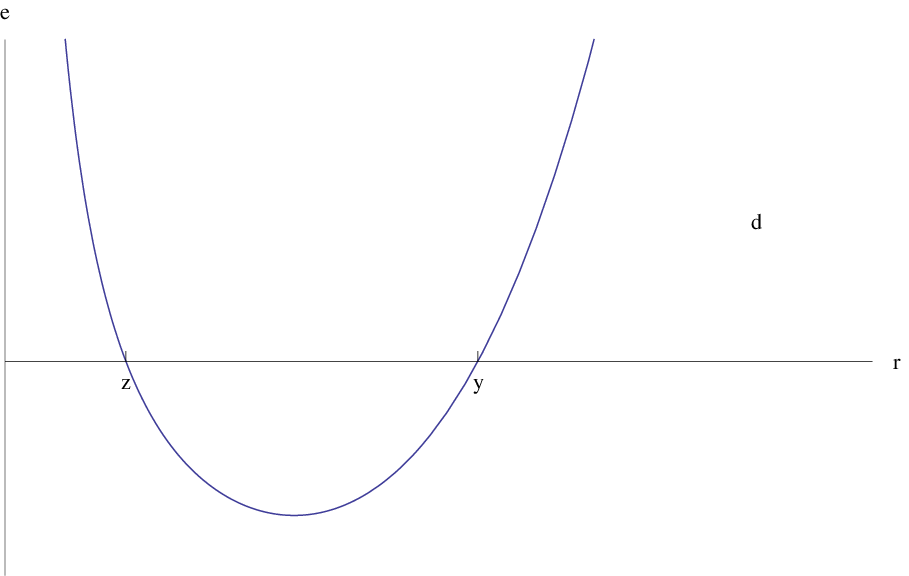}
\end{psfrags}
\end{center}

According to~\eqref{lambda_u}, we have $\partial_u\lambda(0,0)=-\partial_r(1-\mu)(r_+,\varpi_0)<0$. Since $\lambda(0,0)=0$, we may choose $U$ small enough so that $\lambda(u,0)$ is negative for $u\in\,]0,U]$. 
Again we denote by ${\cal P}$ the maximal past set where the solution of the characteristic initial value problem is defined.
In Part~1 we saw that
$\lambda$ is negative on ${\cal P}\setminus\{0\}\times[0,\infty[$, and so, as $\kappa$ is positive (from \eqref{kappa_u} and \eqref{iv}), then $1-\mu$ is negative on ${\cal P}\setminus\{0\}\times[0,\infty[$.

Using the above, we can thus particularize the result of Part 1 on signs and monotonicities to the case where the initial data is\/~\eqref{iu} and\/~\eqref{iv} as follows.
\begin{Lem}[Sign and monotonicity]\label{sign} 
Suppose that $(r,\nu,\lambda,\omega,\theta,\zeta,\kappa)$ is the maximal solution of the characteristic initial value problem\/~\eqref{r_u}$-$\eqref{kappa_at_u}, with initial 
conditions\/~\eqref{iu} and\/~\eqref{iv}.
Then:
\begin{itemize}
\item $\kappa$ is positive;
\item $\nu$ is negative;
\item $\lambda$ is negative on ${\cal P}\setminus\{0\}\times[0,\infty[$;
\item $1 - \mu$ is negative on ${\cal P}\setminus\{0\}\times[0,\infty[$;
\item $r$ is decreasing with both $u$ and $v$;
\item $\omega$ is nondecreasing with both $u$ and $v$.
\end{itemize}
\end{Lem}

Using~\eqref{lambda_u} and~\eqref{theta_u}, we obtain
 \begin{equation}\label{theta_lambda} 
  \partial_u\frac\theta\lambda= 
	-\,\frac\zeta r-\frac\theta\lambda\frac{\nu}{1-\mu}\partial_r(1-\mu),
 \end{equation}
and analogously, using~\eqref{nu_v} and~\eqref{zeta_v},
 \begin{equation}\label{zeta_nu} 
  \partial_v\frac\zeta\nu= 
	-\,\frac\theta r-\frac\zeta\nu\frac{\lambda}{1-\mu}\partial_r(1-\mu).
 \end{equation}

Given $0 < \ckr < r_+$, let us denote by $\cg_{\ckr}$ the level set of the radius function $$\cg_{\ckr}:=\{(u,v)\in{\cal P}:r(u,v)=\ckr\}.$$
If nonempty, $\cg_{\ckr}$ is a connected $C^1$ spacelike curve, since both $\nu$ and $\lambda$ are negative on ${\cal P}\setminus\{0\}\times[0,\infty[$.
Using the Implicit Function Theorem, the facts that $r(0,v)=r_+$,
$r(u,0)=r_+-u$, the signs of $\nu$ and $\lambda$, and the 
breakdown criterion given in Theorem~\ref{bdown},
one can show that $\cg_{\ckr}$ can be parametrized by a $C^1$ function $$v\mapsto (\uckr(v),v),$$
whose domain is $[0,\infty[$ if $\ckr\geq r_+-U$, or an interval of the form $[\vckr(U),\infty[$, for some $\vckr(U)>0$, if $\ckr<r_+-U$\,. Alternatively, $\cg_{\ckr}$ can also be parametrized by a $C^1$ function
$$u\mapsto (u,\vckr(u)),$$ whose domain is always an interval of the form $\,]\uckr(\infty),\min\{r_+-\ckr,U\}]$, for some $\uckr(\infty) \geq 0$\,.
We prove below that if $\ckr>r_-$, then $\uckr(\infty)=0$.

\begin{center}
\begin{turn}{45}
\begin{psfrags}
\psfrag{a}{{\tiny $\vckr(U)$}}
\psfrag{y}{{\tiny $(\ug(v),\tilde v)$}}
\psfrag{d}{{\tiny \!\!$u$}}
\psfrag{u}{{\tiny $v$}}
\psfrag{x}{{\tiny $\ug(v)$}}
\psfrag{f}{{\tiny \!\!\!\!$U$}}
\psfrag{j}{{\tiny \!\!\!\!$\cg_{\ckr}$}}
\psfrag{h}{{\tiny \!\!$\cg_{\ckrp}$}}
\psfrag{v}{{\tiny $u$}}
\psfrag{b}{{\tiny \!\!\!\!\!\!\!\!\!\!\!\!\!\!\!\!$\uckr(\infty)$}}
\psfrag{c}{{\tiny $v$}}
\psfrag{s}{{\tiny \ \ $(u,\vckr(u))$}}
\psfrag{t}{{\tiny \ \ $(\uckr(v),v)$}}
\includegraphics[scale=.8]{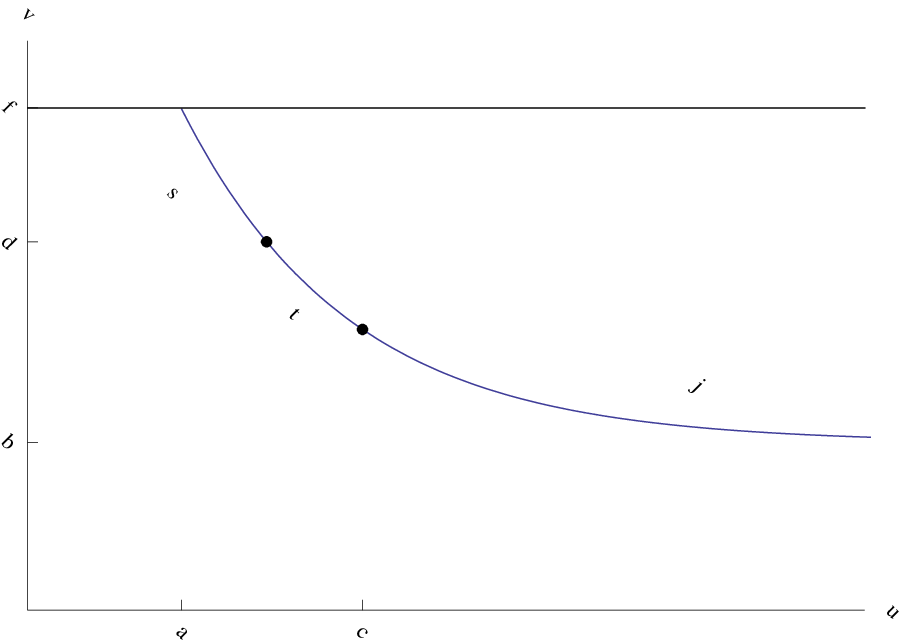}
\end{psfrags}
\end{turn}
\end{center}

To analyze the solution we partition the domain into four regions (see figure below). We start by choosing $\ckrm$ and $\ckrp$ such that
 $r_-<\ckrm<r_0<\ckrp<r_+$. In Section~\ref{red-shift} we treat the region $\ckrp\leq r\leq r_+$.
In Section~\ref{medium} we consider the region $\ckrm\leq r\leq\ckrp$. In Section~\ref{large} we treat the region where $(u,v)$ is such that
$$
\vgr(u)\leq v\leq (1+\beta)\,\vgr(u),
$$
with $\beta>0$ appropriately chosen (we will denote the curve $v=(1+\beta)\,\vgr(u)$ by $\gam$). Finally, in Section~\ref{huge} we consider the region where $(u,v)$ is such that
$$
v\geq (1+\beta)\,\vgr(u).
$$
The reader should regard $\ckrm$, $\ckrp$ and $\beta$ as fixed. Later, they will have to be carefully chosen for our arguments to go through.

\begin{center}
\begin{turn}{45}
\begin{psfrags}
\psfrag{u}{{\tiny $u$}}
\psfrag{v}{{\tiny $v$}}
\psfrag{x}{{\tiny $\!\!\!\cg_{\ckrp}$}}
\psfrag{y}{{\tiny $\!\!\cg_{\ckrm}$}}
\psfrag{z}{{\tiny $\gam$}}
\psfrag{e}{{\tiny $\!\!\!\!\!U$}}
\includegraphics[scale=.7]{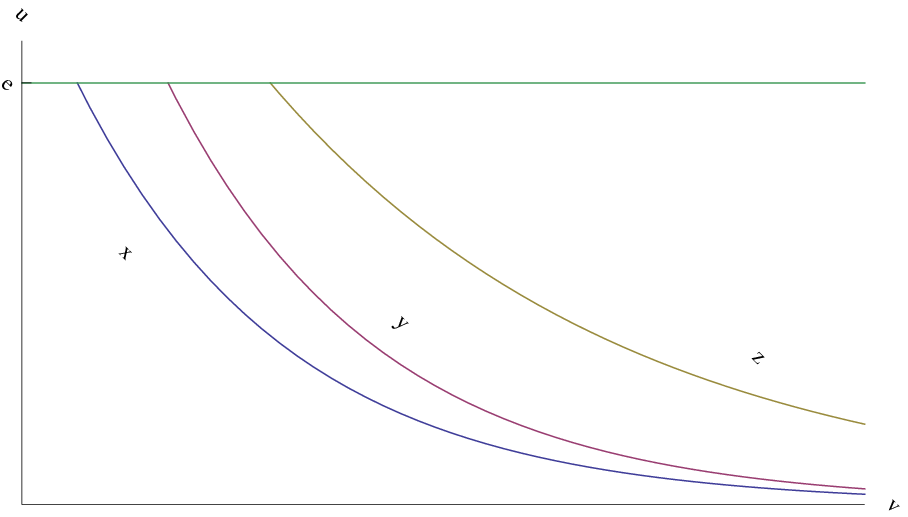}
\end{psfrags}
\end{turn}
\end{center}

The crucial step consists in estimating the fields $\frac\theta\lambda$ and $\frac\zeta\nu$. Once this is done, the other estimates follow easily.
By integrating~\eqref{theta_lambda} and~\eqref{zeta_nu}, we obtain
\begin{eqnarray} 
\frac\theta\lambda(u,v)&=&\frac\theta\lambda(\uckr(v),v)e^{-\int_{\uckr(v)}^{u}\bigl[\frac{\nu}{1-\mu}\partial_r(1-\mu)\bigr](\tilde u,v)\,d\tilde u}\nonumber\\
&&-\int_{\uckr(v)}^{u}\frac\zeta r(\tilde u,v)e^{-\int_{\tilde u}^{u}\bigl[\frac{\nu}{1-\mu}\partial_r(1-\mu)\bigr](\bar{u},v)\,d\bar{u}}\,d\tilde u,\label{field_6}\\
\frac\zeta\nu(u,v)&=&\frac\zeta\nu(u,\vckr(u))e^{-\int_{\vckr(u)}^{v}\bigl[\frac{\lambda}{1-\mu}\partial_r(1-\mu)\bigr](u,\tilde v)\,d\tilde v}\nonumber\\
&&-\int_{\vckr(u)}^{v}\frac\theta r(u,\tilde v)e^{-\int_{\tilde v}^{v}\bigl[\frac{\lambda}{1-\mu}\partial_r(1-\mu)\bigr](u,\bar{v})\,d\bar{v}}\,d\tilde v.\label{field_7}
\end{eqnarray}

Formula~\eqref{field_6} is valid provided that $\uckr(v)$ is defined and $\uckr(v)\leq u$, since the domain ${\cal P}$ is a past set; it also holds if we replace $\uckr(v)$ by $0$.
Similarly, formula~\eqref{field_7} is valid provided that $\vckr(u)$ is defined and $\vckr(u)\leq v$; again it holds if we replace $\vckr(u)$ by $0$.

\section{The region $J^-(\cg_{\ckrp})$}\label{red-shift}

Recall that $r_0<\ckrp<r_+$.
In this section, we treat the region $\ckrp\leq r\leq r_+$, that is, $J^-(\cg_{\ckrp})$.\footnote{Throughout this paper we follow the usual notations for the causal structure of the quotient Lorentzian manifold with coordinates $(u,v)$ and time orientation such that $\frac{\partial}{\partial u}$ and $\frac{\partial}{\partial v}$ are future-pointing.} 
Our first goal is to estimate~\eqref{z_00} for $\frac\zeta\nu$. 
This will allow us to obtain the lower bound~\eqref{k_min} for $\kappa$,
which will then be used to improve estimate~\eqref{z_00} to~\eqref{z_0_bis}.
Finally, we successively bound $\frac\theta\lambda$, $\theta$, $\varpi$, and use this to prove that the domain of
$\vgrp(\,\cdot\,)$ is $\,]0,\min\{r_+-\ckrp,U\}]$.

In this region, the solution with general $\zeta_0$ can then be considered as a small perturbation of the 
Reissner-Nordstr\"{o}m solution~\eqref{primeira}$-$\eqref{ultima}: $\omega$ is close to $\omega_0$, $\kappa$ is close to $1$ and $\zeta, \theta$ are
close to $0$. Besides, the smaller $U$ is, the closer the approximation.

Substituting~\eqref{field_6} in~\eqref{field_7}  (with both $\uckr(v)$ and $\vckr(u)$ replaced by 0), we get
\begin{eqnarray}
\frac\zeta\nu(u,v)&=&
\frac\zeta\nu(u,0)e^{-\int_{0}^{v}\bigl[\frac{\lambda}{1-\mu}\partial_r(1-\mu)\bigr](u,\tilde v)\,d\tilde v}\label{double_zero}\\
&&+\int_{0}^{v}\frac\theta\lambda(0,\tilde v)e^{-\int_{0}^{u}\bigl[\frac{\nu}{1-\mu}\partial_r(1-\mu)\bigr](\tilde u,\tilde v)\,d\tilde u}\times\nonumber\\
&&\qquad\qquad\times
\Bigl[\frac{(-\lambda)}{r}\Bigr](u,\tilde v)e^{-\int_{\tilde v}^{v}\bigl[\frac{\lambda}{1-\mu}\partial_r(1-\mu)\bigr](u,\bar{v})\,d\bar{v}}\,d\tilde v\nonumber\\
&&+\int_{0}^{v}\left(\int_{0}^{u}\Bigl[\frac\zeta\nu\frac{(-\nu)}r\Bigr](\tilde u,\tilde v)
e^{-\int_{\tilde u}^{u}\bigl[\frac{\nu}{1-\mu}\partial_r(1-\mu)\bigr](\bar{u},\tilde v)\,d\bar{u}}\,d\tilde u\right)\times\nonumber\\
&&\qquad\qquad\times\Bigl[\frac{(-\lambda)}{r}\Bigr](u,\tilde v)
e^{-\int_{\tilde v}^{v}\bigl[\frac{\lambda}{1-\mu}\partial_r(1-\mu)\bigr](u,\bar{v})\,d\bar{v}}\,d\tilde v.\nonumber
\end{eqnarray}
We make the change of coordinates 
\begin{equation}
(u,v)\mapsto(r(u,v),v)\quad \Leftrightarrow\quad (r,v)\mapsto (\ug(v),v).
\label{coordinates}
\end{equation}
The coordinates $(r,v)$ are called Bondi coordinates.
We denote by $\hatz$ the function $\frac\zeta\nu$ written in these new coordinates, so that
$$
\frac\zeta\nu(u,v)=\hatz(r(u,v),v)
\quad \Leftrightarrow\quad
\hatz(r,v)=\frac\zeta\nu(\ug(v),v).
$$
The same notation will be used for other functions.
In the new coordinates, \eqref{double_zero} may be written 
\begin{eqnarray}
\hatz(r,v)&=&
\frac\zeta\nu(\ug(v),0)e^{-\int_{0}^{v}\bigl[\frac{\lambda}{1-\mu}\partial_r(1-\mu)\bigr](\ug(v),\tilde v)\,d\tilde v}\label{triple_zero}\\
&&+\int_{0}^{v}\widehat{\frac\theta\lambda}(r_+,\tilde v)
e^{\int_{r(\ug(v),\tilde v)}^{r_+}\bigl[\frac{1}{\widehat{1-\mu}}\widehat{\partial_r(1-\mu)}\bigr](\tilde{s},\tilde v)\,d\tilde{s}}\times\nonumber\\
&&\qquad\qquad\times
\Bigl[\frac{(-\lambda)(\ug(v),\tilde v)}{r(\ug(v),\tilde v)}\Bigr]e^{-\int_{\tilde v}^{v}\bigl[\frac{\lambda}{1-\mu}\partial_r(1-\mu)\bigr](\ug(v),\bar{v})\,d\bar{v}}\,d\tilde v\nonumber\\
&&+\int_{0}^{v}\left(\int_{r(\ug(v),\tilde v)}^{r_+}\frac{1}{\tilde s}\Bigl[\hatz\Bigr](\tilde s,\tilde v)
e^{\int_{r(\ug(v),\tilde v)}^{\tilde s}\bigl[\frac{1}{\widehat{1-\mu}}\widehat{\partial_r(1-\mu)}\bigr](\bar{s},\tilde v)\,d\bar{s}}\,d\tilde s\right)\times\nonumber\\
&&\qquad\qquad\times
\Bigl[\frac{(-\lambda)(\ug(v),\tilde v)}{r(\ug(v),\tilde v)}\Bigr]
e^{-\int_{\tilde v}^{v}\bigl[\frac{\lambda}{1-\mu}\partial_r(1-\mu)\bigr](\ug(v),\bar{v})\,d\bar{v}}\,d\tilde v.\nonumber
\end{eqnarray}

\begin{center}
\begin{turn}{45}
\begin{psfrags}
\psfrag{a}{{\tiny $\tilde v$}}
\psfrag{y}{{\tiny $(\ug(v),\tilde v)$}}
\psfrag{u}{{\tiny $v$}}
\psfrag{z}{{\tiny \!\!\!\!$r_+-s_1$}}
\psfrag{x}{{\tiny $\ug(v)$}}
\psfrag{f}{{\tiny $U$}}
\psfrag{g}{{\tiny $\cg_{s_2}$}}
\psfrag{h}{{\tiny $\cg_{s_1}$}}
\psfrag{v}{{\tiny $u$}}
\psfrag{b}{{\tiny $v$}}
\psfrag{c}{{\tiny $0$}}
\psfrag{p}{{\tiny $(\ug(v),v)$}}
\includegraphics[scale=.8]{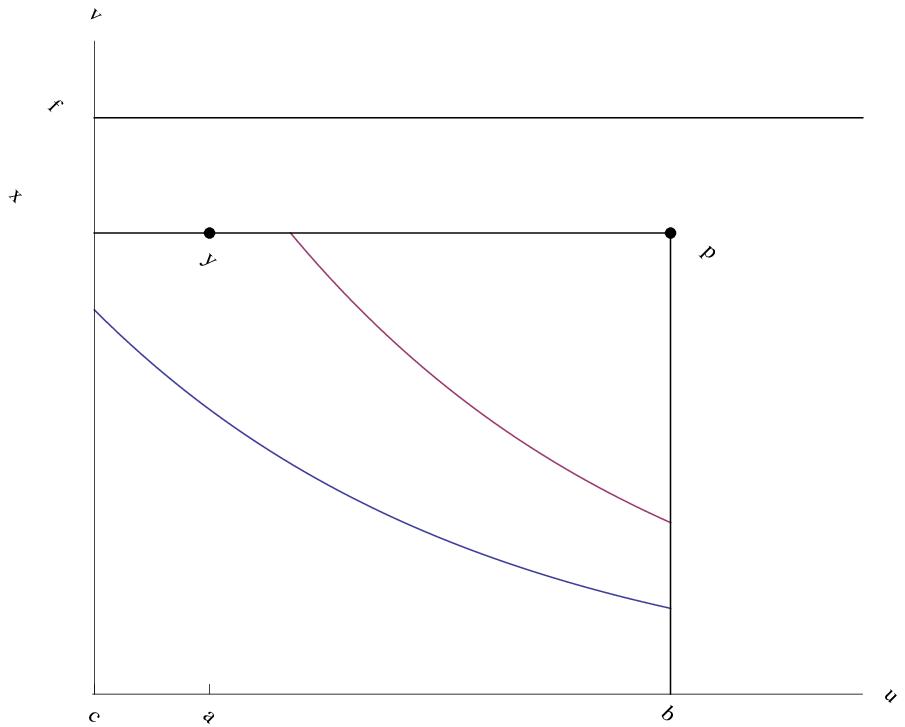}
\end{psfrags}
\end{turn}
\end{center}

We have $\theta(0,v)=0$ and, from~\eqref{theta_u}, $\partial_u\theta(0,v)=0$, whereas $\lambda(0,v)=0$ and, from~\eqref{lambda_u}, $\partial_u\lambda(0,v)<0$. Writing
\[
\frac{\theta}{\lambda} (u,v) = \frac{\int_0^u \partial_u\theta(\tilde{u},v) d\tilde{u}}{\int_0^u \partial_u\lambda(\tilde{u},v) d\tilde{u}},
\]
it is easy to show that the function $\frac\theta\lambda$ can be extended as a continuous function to $\{0\}\times[0,\infty[$, with $\frac\theta\lambda(0,v)=0$.
Substituting this into~\eqref{field_6} (again with $\uckr(v)$ replaced by 0) yields
\begin{eqnarray*}
\frac\theta\lambda(u,v) 
&=&-\int_{0}^{u}\frac\zeta r(\tilde u,v)e^{-\int_{\tilde u}^{u}\bigl[\frac{\nu}{1-\mu}\partial_r(1-\mu)\bigr](\bar{u},v)\,d\bar{u}}\,d\tilde u.
\end{eqnarray*}
We can rewrite this in the new coordinates as
\begin{eqnarray}
\widehat{\frac\theta\lambda}(r,v)&=&\int_{r}^{r_+}\frac 1 {\tilde s}\Bigl[\hatz\Bigr](\tilde s,v)
e^{\int_r^{\tilde s}\bigl[\frac{1}{\widehat{1-\mu}}\widehat{\partial_r(1-\mu)}\bigr](\bar{s},v)\,d\bar{s}}\,d\tilde s.\label{t_z}
\end{eqnarray}

A key point is to bound the exponentials that appear in~\eqref{triple_zero} and~\eqref{t_z}.
As we go on, this will be done several times in different ways.

\begin{Lem}\label{run}
 Assume that there exists $\alpha \geq 0$ such that, for $0 \leq \tilde{v} \leq v$, the following bounds hold:
$$
e^{-\int_{\tilde v}^{v}\bigl[\frac{\lambda}{1-\mu}\partial_r(1-\mu)\bigr](\ug(v),\bar{v})\,d\bar{v}}\leq e^{-\alpha(v-\tilde v)}
$$
and
$$
e^{\int_{r(\ug(v),\tilde v)}^{\tilde s}\bigl[\frac{1}{\widehat{1-\mu}}\widehat{\partial_r(1-\mu)}\bigr](\bar{s},\tilde v)\,d\bar{s}}\leq 1.
$$
Then~\eqref{triple_zero} implies
\begin{equation}\label{z_0}
\Bigl|\hatz\Bigr|(r,v)\leq e^{\frac{(r_+-r)^2}{rr_+}}\max_{u\in[0,\ug(v)]}|\zeta_0|(u)e^{-\alpha v}.
\end{equation}
\end{Lem}
\begin{proof}
Combining~\eqref{triple_zero} with $\widehat{\frac\theta\lambda}(r_+,v)\equiv 0$ and the bounds on the exponentials, we have
\begin{eqnarray}
\Bigl|\hatz\Bigr|(r,v)&\leq&|\zeta_0|(\ug(v))e^{-\alpha v}\label{four_zero}\\
&&+\int_{0}^{v}\int_{r(\ug(v),\tilde v)}^{r_+}\frac{1}{\tilde s}\Bigl[\hatz\Bigr](\tilde s,\tilde v)
\,d\tilde s\times\nonumber\\
&&\qquad\qquad\qquad\times
\Bigl[\frac{(-\lambda)(\ug(v),\tilde v)}{r(\ug(v),\tilde v)}\Bigr]
e^{-\alpha(v-\tilde v)}\,d\tilde v.\nonumber
\end{eqnarray}
For $r\leq s< r_+$, define
$$
{\cal Z}_{(r,v)}^\alpha(s)=\left\{
\begin{array}{ll}
\max_{\tilde v\in[0,v]}\left\{e^{\alpha\tilde v}\Bigl|\hatz\Bigr|(s,\tilde v)\right\}&{\rm if}\ r_+-\ug(v)\leq s<r_+,\\
\max_{\tilde v\in[\vgs(\ug(v)),v]}\left\{e^{\alpha\tilde v}\Bigl|\hatz\Bigr|(s,\tilde v)\right\}&{\rm if}\ r\leq s\leq r_+-\ug(v).
\end{array}
\right.
$$
Here the maximum is taken over the projection of $J^-(\ug(v),v)\cap\cg_s$ on the $v$-axis (see the figure below).

\begin{center}
\begin{turn}{45}
\begin{psfrags}
\psfrag{a}{{\tiny $\vgsd(\ug(v))$}}
\psfrag{y}{{\tiny $(\ug(v),\vgsd(\ug(v)))$}}
\psfrag{u}{{\tiny $v$}}
\psfrag{z}{{\tiny $r_+-s_1$}}
\psfrag{x}{{\tiny $\ug(v)$}}
\psfrag{f}{{\tiny $U$}}
\psfrag{g}{{\tiny $\cg_{s_2}$}}
\psfrag{h}{{\tiny $\cg_{s_1}$}}
\psfrag{v}{{\tiny $u$}}
\psfrag{b}{{\tiny $v$}}
\psfrag{c}{{\tiny $0$}}
\psfrag{p}{{\tiny $(\ug(v),v)$}}
\includegraphics[scale=.8]{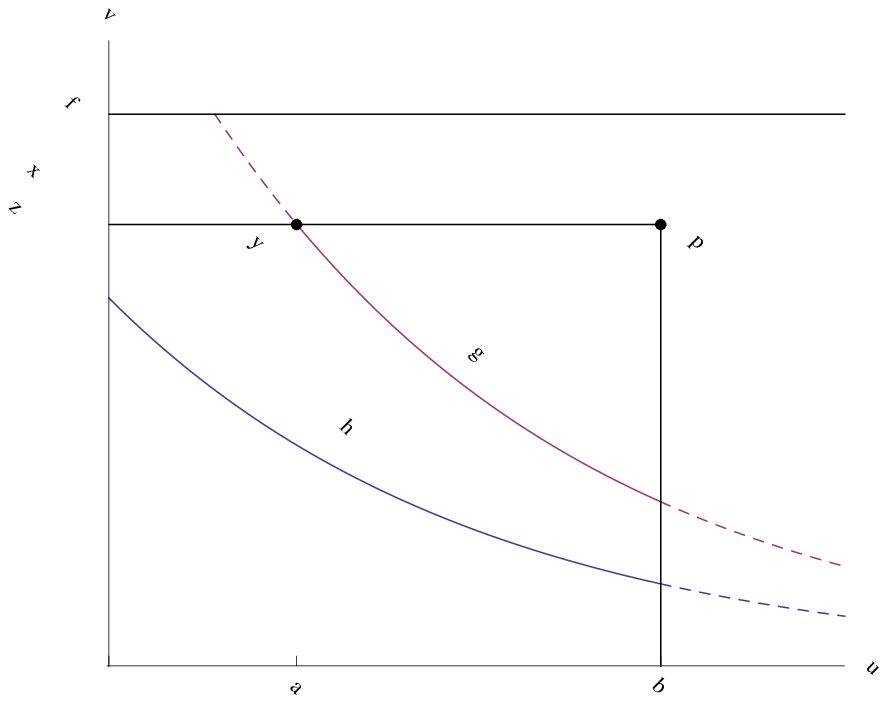}
\end{psfrags}
\end{turn}
\end{center}

Note that ${\cal Z}_{(r,v)}^\alpha(r)=e^{\alpha v}\bigl|\hatz\bigr|(r,v)$. Inequality~\eqref{four_zero} implies
\begin{eqnarray}
{\cal Z}_{(r,v)}^\alpha(r)&\leq& |\zeta_0|(\ug(v))\nonumber\\
&&+\int_{0}^{v}\int_{r(\ug(v),\tilde v)}^{r_+}{\cal Z}_{(r,v)}^\alpha(\tilde s)
\,d\tilde s
\Bigl[\frac{(-\lambda)(\ug(v),\tilde v)}{[r(\ug(v),\tilde v)]^2}\Bigr]
\,d\tilde v\nonumber\\
&\leq& \max_{s\in[r,r_+]}|\zeta_0|(\ugs(v))\nonumber\\
&&+\int_{r}^{r_+}{\cal Z}_{(r,v)}^\alpha(\tilde s)
\,d\tilde s\int_{0}^{v}
\Bigl[\frac{(-\lambda)(\ug(v),\tilde v)}{[r(\ug(v),\tilde v)]^2}\Bigr]
\,d\tilde v\nonumber\\
&\leq& \max_{u\in[0,\ug(v)]}|\zeta_0|(u)
+
\Bigl(\frac 1r-\frac 1{r_+-\ug(v)}\Bigr)
\int_{r}^{r_+}{\cal Z}_{(r,v)}^\alpha(\tilde s)
\,d\tilde s.\nonumber
\end{eqnarray}
We still consider $r\leq s <r_+$.
Let $\tilde v\in[0,v]$ if $r_+-\ug(v)\leq s<r_+$, 
and $\tilde v\in [\vgs(\ug(v)),v]$ if $r\leq s\leq r_+-\ug(v)$.
In this way $(\ugs(\tilde v),\tilde v)\in J^-(\ug(v),v)$.
In the same way one can show that
$$
e^{\alpha\tilde v}\Bigl|\hatz\Bigr|(s,\tilde v)\leq \max_{u\in[0,\ugs(v)]}|\zeta_0|(u)
+
\Bigl(\frac 1s-\frac 1{r_+-\ugs(v)}\Bigr)
\int_{s}^{r_+}{\cal Z}_{(r,v)}^\alpha(\tilde s)
\,d\tilde s
$$
because $J^-(\ugs(\tilde v),\tilde v)\cap \cg_{\tilde s}\subset J^-(\ug(v),v)\cap \cg_{\tilde s}$, for $s\leq\tilde s<r_+$, and so
${\cal Z}_{(s,\tilde v)}^\alpha(\tilde s)\leq {\cal Z}_{(r,v)}^\alpha(\tilde s)$.
Since $ \ugs(v)\leq \ug(v)$ for $r\leq s<r_+$, we have
\begin{eqnarray*}
{\cal Z}_{(r,v)}^\alpha(s)
&\leq& \max_{u\in[0,\ug(v)]}|\zeta_0|(u)
+
\Bigl(\frac 1r-\frac 1{r_+}\Bigr)
\int_{s}^{r_+}{\cal Z}_{(r,v)}^\alpha(\tilde s)
\,d\tilde s.
\end{eqnarray*}
Using Gronwall's inequality, we get
$$
{\cal Z}_{(r,v)}^\alpha(r)\leq e^{\frac{(r_+-r)^2}{rr_+}}\max_{u\in[0,\ug(v)]}|\zeta_0|(u).
$$
This establishes~\eqref{z_0}.
\end{proof}

\begin{Lem}\label{lso}
 Let $r_0\leq r<r_+$ and $v>0$. Then
\begin{equation}\label{z_00}
\Bigl|\hatz\Bigr|(r,v)\leq e^{\frac{(r_+-r)^2}{rr_+}}\max_{u\in[0,\ug(v)]}|\zeta_0|(u).
\end{equation}
\end{Lem}
\begin{proof}
We bound the exponentials in~\eqref{triple_zero}. 
From~\eqref{eta}, the definition of $\eta$ and $\varpi\geq\varpi_0$,
$$
-\partial_r(1-\mu)=\frac{2\eta}{r^2}\leq\frac{2\eta_0}{r^2}=-\partial_r(1-\mu)(r,\varpi_0)\leq 0.
$$
Therefore in the region $J^-(\cg_{r_0})$ the exponentials are bounded by 1.
Applying Lemma~\ref{run} with $\alpha=0$ we obtain~\eqref{z_0} with $\alpha=0$, which is precisely~\eqref{z_00}.
\end{proof}

According to~\eqref{z_00}, the function $\frac\zeta\nu$ is bounded in the region $J^-(\cg_{r_0})$, say by $\hat\delta$.
From~\eqref{kappa_u},
\begin{eqnarray}
\kappa(u,v)&=&e^{\int_0^u\bigl(\frac{\zeta^2}{r\nu}\bigr)(\tilde u,v)\,d\tilde u}\nonumber\\
&\geq&e^{\hat\delta^2\int_0^u(\frac{\nu}{r})(\tilde u,v)\,d\tilde u}\nonumber\\
&\geq&\left(\frac{r_0}{r_+}\right)^{\hat\delta^2}.\label{k_min}
\end{eqnarray}

We recall from Part 1 that equations (\ref{r_v}), (\ref{nu_v}), (\ref{omega_v}) and (\ref{kappa_at_u}) imply
\begin{equation}\label{Ray} 
\partial_v\left(\frac{1-\mu}{\nu}\right)=-\,\frac{\theta^2}{\nu r\kappa},
\end{equation}
which is the Raychaudhuri equation in the $v$ direction. We also recall that the integrated form of~\eqref{omega_u} is
\begin{eqnarray} 
\omega(u,v)=\omega_0(v)e^{-\int_0^u\bigl(\frac{\zeta^2}{r\nu}\bigr)(u',v)\,du'}\qquad\qquad\qquad\qquad\qquad\qquad\qquad&&\nonumber\\ \qquad\qquad\ \  +
\int_0^ue^{-\int_s^u\frac{\zeta^2}{r\nu}(u',v)\,du'}\left(\frac{1}{2}\left(1+\frac{e^2}{r^2}
-\frac{\Lambda}{3}r^2\right)\frac{\zeta^2}{\nu}\right)(s,v)\,ds.&&\label{omega_final}
\end{eqnarray}
These will be used in the proof of the following result.

\begin{Prop}\label{psd}
  Let $r_0<\ckrp\leq r<r_+$ and $v>0$. Then there exists $\alpha>0$ (given by~\eqref{alp_new} below) such that
\begin{equation}\label{z_0_bis}
\Bigl|\hatz\Bigr|(r,v)\leq e^{\frac{(r_+-r)^2}{rr_+}}\max_{u\in[0,\ug(v)]}|\zeta_0|(u)e^{-\alpha v},
\end{equation}
\begin{equation} 
\Bigl|\widehat{\frac\theta\lambda}\Bigr|(r,v) 
\leq \hat C_r\max_{u\in[0,\ug(v)]}|\zeta_0|(u)e^{-\alpha v},\label{t_0}
\end{equation}
\begin{equation}
|\hat\theta|(r,v)\leq C\max_{u\in[0,\ug(v)]}|\zeta_0|(u)e^{-\alpha v}.\label{theta_f}
\end{equation}
For $(u,v)\in J^-(\cg_{\ckrp})$, and $U$ sufficiently small, we have
\begin{equation}\label{mass_2}
\varpi_0\leq\varpi(u,v)\leq \varpi_0+
C\left(\sup_{\tilde u\in[0,u]}|\zeta_0|(\tilde u)\right)^2.
\end{equation}
Moreover, the curve $\cg_{\ckrp}$ intersects every line of constant $u$ provided that $0 < u \leq \min\{r_+-\ckrp,U\}$. Therefore, $\ugrp(\infty)=0$.
\end{Prop}
\begin{proof}
In $J^-(\cg_{\ckrp})$, we have
$\partial_r(1-\mu)(r,\varpi_0)
\geq\min_{r\in[\ckrp,r_+]}\partial_r(1-\mu)(r,\varpi_0)>0$ and
\begin{eqnarray} 
-\,\frac{\lambda}{1-\mu}\partial_r(1-\mu)&\leq&-\kappa\, \partial_r(1-\mu)(r,\varpi_0)\nonumber\\
&\leq &-\inf_{J^-(\cg_{\ckrp})}\kappa\times\partial_r(1-\mu)(r,\varpi_0)\nonumber\\
&\leq &
-\Bigl(\frac{\ckrp}{r_+}\Bigr)^{\hat\delta^2}\min_{r\in[\ckrp,r_+]}\partial_r(1-\mu)(r,\varpi_0)\nonumber\\
&=:&-\alpha<0,\label{alp_new}
\end{eqnarray}
where we have used~\eqref{k_min} (with $\ckrp$ instead of $r_0$).
Thus, we can improve the bounds on the exponentials in~\eqref{triple_zero} that involve integrals in $v$ as follows:
$$
e^{-\int_{\tilde v}^{v}\bigl[\frac{\lambda}{1-\mu}\partial_r(1-\mu)\bigr](\ug(v),\bar{v})\,d\bar{v}}\leq e^{-\alpha(v-\tilde v)}.
$$
Since 
$$
\frac{1}{\widehat{1-\mu}}\widehat{\partial_r(1-\mu)}\leq 0,
$$
as before, we have
$$
e^{\int_{r(\ug(v),\tilde v)}^{\tilde s}\bigl[\frac{1}{\widehat{1-\mu}}\widehat{\partial_r(1-\mu)}\bigr](\bar{s},\tilde v)\,d\bar{s}}\leq 1.
$$
We apply Lemma~\ref{run} again, this time with a positive $\alpha$, to get~\eqref{z_0_bis}.

Now we may use~\eqref{t_z} and~\eqref{z_0_bis} to obtain
\begin{eqnarray*}
\Bigl|\widehat{\frac\theta\lambda}\Bigr|(r,v)&\leq& e^{\frac{(r_+-r)^2}{rr_+}}\ln\Bigl(\frac{r_+}r\Bigr)\max_{u\in[0,\ug(v)]}|\zeta_0|(u)e^{-\alpha v}\nonumber\\
&=&\hat C_r\max_{u\in[0,\ug(v)]}|\zeta_0|(u)e^{-\alpha v}.
\end{eqnarray*}

In order to bound $\varpi$ in $J^-(\cg_{\ckrp})$,
we note that
$$
-\int_0^{u}\left(\frac{\zeta^2}{r\nu}\right)(\tilde u,v)\,d\tilde u\leq C^2\left(\sup_{\tilde u\in[0,u]}|\zeta_0|(\tilde u)\right)^2\ln\left(\frac{r_+}{\ckrp}\right),
$$
$$
\left|1+\frac{e^2}{r^2}-\frac\Lambda 3r^2\right|\leq 1+\frac{e^2}{\ckrp^2}+\frac{|\Lambda|}{3}r_+^2
$$
and
$$
-\int_0^{u}\nu(\tilde u,v)\,d\tilde u=r_+-r(u,v)\leq r_+-\ckrp.
$$
From~\eqref{omega_final}, we conclude that 
$$ 
\varpi(u,v)\leq \varpi_0e^{C\left(\sup_{\tilde u\in[0,u]}|\zeta_0|(\tilde u)\right)^2}+
C\left(\sup_{\tilde u\in[0,u]}|\zeta_0|(\tilde u)\right)^2.
$$ 
Inequality~\eqref{mass_2} follows from $e^x\leq 1+2x$, for small $x$, since
$\zeta_0$ is continuous and $\zeta_0(0)=0$.

Given that $\kappa\leq 1$, we have $(1-\mu)\leq\lambda$. Moreover, since $\varpi$ is bounded
in the region $J^-(\cg_{\ckrp})$, $1-\mu$ is bounded from below, and so $\lambda$ is also bounded from below.
Hence~\eqref{t_0} implies~\eqref{theta_f}.

Let $0<u\leq\min\{r_+-\ckrp,U\}$. We claim that
\begin{equation}\label{stop}
\sup\left\{v\in[0,\infty[:(u,v)\in J^-(\cg_{\ckrp})\right\}<\infty.
\end{equation}
To see this, first note that~\eqref{nu_v} shows that $v\mapsto\nu(u,v)$ is decreasing in $J^-(\cg_{\ckrp})$, as $\partial_r(1-\mu)\geq 0$ for $r_0 \leq r \leq r_+$ (recall that $\eta(r,\omega) \leq \eta_0(r)$).
Then~\eqref{Ray} shows $v\mapsto(1-\mu)(u,v)$ is also decreasing in $J^-(\cg_{\ckrp})$. Thus, as long as $v$ 
is such that $(u,v)\in J^-(\cg_{\ckrp})$, we have
$(1-\mu)(u,v)\leq (1-\mu)(u,0)<0$.
Combining the previous inequalities with~\eqref{k_min}, we get
$$
\lambda(u,v)\leq \left(\frac{\ckrp}{r_+}\right)^{\hat\delta^2}(1-\mu)(u,0)<0.
$$
Finally, if \eqref{stop} did not hold for a given $u$, we would have
\begin{eqnarray*}
r(u,v)&=&r(u,0)+\int_0^v\lambda(u,v')\,dv'\\ &\leq& r(u,0)+\left(\frac{\ckrp}{r_+}\right)^{\hat\delta^2}(1-\mu)(u,0)\,v\to-\infty,
\end{eqnarray*}
as $v \to \infty$, which is a contradiction. This establishes the claim. 
\end{proof}

\section{The region $J^-(\cg_{\ckrm})\cap J^+(\cg_{\ckrp})$}\label{medium}

In this section, we treat the region $\ckrm\leq r\leq\ckrp$. Recall that we assume that $r_-<\ckrm<r_0<\ckrp<r_+$. By decreasing $\ckrm$, if necessary, we will also assume that 
\begin{equation}\label{assume}
-(1-\mu)(\ckrm,\varpi_0) \leq - (1-\mu)(\ckrp,\varpi_0).
\end{equation}

In Subsection~\eqref{medium-1}, we obtain estimates~\eqref{z-0} and~\eqref{t-0-a} for $\frac\zeta\nu$ and
$\frac\theta\lambda$, which will allow us to obtain the lower bound~\eqref{k_min_2} for $\kappa$, the upper bound~\eqref{mass_2_bis} for $\varpi$,
and to prove that the domain of
$\vgr(\,\cdot\,)$ is $\,]0,\min\{r_+-\ckrm,U\}]$.
In Subsection~\eqref{nu-lambda}, we obtain upper and lower bounds for $\lambda$ and $\nu$, as well as more information about the region $\ckrm\leq r\leq\ckrp$.
In Subsection~\eqref{medium-3}, we use the results from the previous subsection
to improve the estimates on $\frac\zeta\nu$ and $\frac\theta\lambda$ to~\eqref{z_f} and~\eqref{t-f}.
We also obtain the bound~\eqref{theta_ff} for $\theta$.

As in the previous section, the solution with general $\zeta_0$ is qualitatively still a small perturbation of the Reissner-Nordstr\"{o}m solution~\eqref{primeira}$-$\eqref{ultima}: $\omega$, $\kappa$, $\zeta$ and $\theta$ remain close to $\omega_0$, $1$ and $0$, respectively. Moreover, $\lambda$ is bounded from below by a negative constant, and away from zero by a constant depending on $\ckrp$ and $\ckrm$, as is also the case in the Reissner-Nordstr\"{o}m solution (see equation~\eqref{primeira}). Likewise, $\nu$ has a similar behavior to its Reissner-Nordstr\"{o}m counterpart (see equation~\eqref{segunda}): when multiplied by $u$, $\nu$ behaves essentially like $\lambda$.

\subsection{First estimates}\label{medium-1}
By reducing $U>0$, if necessary, we can assume  $U\leq r_+-\ckrp$.
We turn our attention to the region $J^-(\cg_{\ckrm})\cap J^+(\cg_{\ckrp})$.
Substituting~\eqref{field_6} in~\eqref{field_7} with $\ckr=\ckrp$, we get
\begin{eqnarray} 
\frac\zeta\nu(u,v)&=&
\frac\zeta\nu(u,\vgrp(u))e^{-\int_{\vgrp(u)}^{v}\bigl[\frac{\lambda}{1-\mu}\partial_r(1-\mu)\bigr](u,\tilde v)\,d\tilde v}\label{double}\\
&&+\int_{\vgrp(u)}^{v}\frac\theta\lambda(\ugrp(\tilde v),\tilde v)e^{-\int_{\ugrp(\tilde v)}^{u}
\bigl[\frac{\nu}{1-\mu}\partial_r(1-\mu)\bigr](\tilde u,\tilde v)\,d\tilde u}\times\nonumber\\
&&\qquad\qquad\times
\Bigl[\frac{(-\lambda)}{r}\Bigr](u,\tilde v)e^{-\int_{\tilde v}^{v}\bigl[\frac{\lambda}{1-\mu}\partial_r(1-\mu)\bigr](u,\bar{v})\,d\bar{v}}\,d\tilde v\nonumber\\
&&+\int_{\vgrp(u)}^{v}\left(\int_{\ugrp(\tilde v)}^{u}\Bigl[\frac\zeta\nu\frac{(-\nu)}r\Bigr](\tilde u,\tilde v)
e^{-\int_{\tilde u}^{u}\bigl[\frac{\nu}{1-\mu}\partial_r(1-\mu)\bigr](\bar{u},\tilde v)\,d\bar{u}}\,d\tilde u\right)\times\nonumber\\
&&\qquad\qquad\times\Bigl[\frac{(-\lambda)}{r}\Bigr](u,\tilde v)
e^{-\int_{\tilde v}^{v}\bigl[\frac{\lambda}{1-\mu}\partial_r(1-\mu)\bigr](u,\bar{v})\,d\bar{v}}\,d\tilde v.\nonumber
\end{eqnarray}
We make the change of coordinates~\eqref{coordinates}.
Then, \eqref{double} may be written
\begin{eqnarray} 
\hatz(r,v)&=&
\hatz(\ckrp,\vgrp(\ug(v)))e^{-\int_{\vgrp(\ug(v))}^{v}\bigl[\frac{\lambda}{1-\mu}\partial_r(1-\mu)\bigr](\ug(v),\tilde v)\,d\tilde v}\label{triple}\\
&&+\int_{\vgrp(\ug(v))}^{v}\widehat{\frac\theta\lambda}(\ckrp,\tilde v)
e^{\int_{r(\ug(v),\tilde v)}^{\ckrp}\bigl[\frac{1}{\widehat{1-\mu}}\widehat{\partial_r(1-\mu)}\bigr](\tilde{s},\tilde v)\,d\tilde{s}}\times\nonumber\\
&&\qquad\qquad\times
\Bigl[\frac{(-\lambda)(\ug(v),\tilde v)}{r(\ug(v),\tilde v)}\Bigr]e^{-\int_{\tilde v}^{v}\bigl[\frac{\lambda}{1-\mu}\partial_r(1-\mu)\bigr](\ug(v),\bar{v})\,d\bar{v}}\,d\tilde v\nonumber\\
&&+\int_{\vgrp(\ug(v))}^{v}\left(\int_{r(\ug(v),\tilde v)}^{\ckrp}\Bigl[\hatz\frac{1}{\tilde s}\Bigr](\tilde s,\tilde v)
e^{\int_{r(\ug(v),\tilde v)}^{\tilde s}\bigl[\frac{1}{\widehat{1-\mu}}\widehat{\partial_r(1-\mu)}\bigr](\bar{s},\tilde v)\,d\bar{s}}\,d\tilde s\right)\times\nonumber\\
&&\qquad\qquad\times
\Bigl[\frac{(-\lambda)(\ug(v),\tilde v)}{r(\ug(v),\tilde v)}\Bigr]
e^{-\int_{\tilde v}^{v}\bigl[\frac{\lambda}{1-\mu}\partial_r(1-\mu)\bigr](\ug(v),\bar{v})\,d\bar{v}}\,d\tilde v.\nonumber
\end{eqnarray}
For $(r,v)$ such that $(\ug(v),v)\in J^-(\cg_{\ckrm})\cap J^+(\cg_{\ckrp})$, $\vgrp(\ug(v))$ is well defined because $U\leq r_+-\ckrp$.

\begin{center}
\begin{turn}{45}
\begin{psfrags}
\psfrag{a}{{\tiny $\tilde v$}}
\psfrag{y}{{\tiny $(\ug(v),\tilde v)$}}
\psfrag{d}{{\tiny $\vgrp(\ug(v))$}}
\psfrag{u}{{\tiny $v$}}
\psfrag{z}{{\tiny \!\!\!\!$r_+-s_1$}}
\psfrag{x}{{\tiny $\ug(v)$}}
\psfrag{f}{{\tiny $U=r_+-\ckrp$}}
\psfrag{g}{{\tiny \!\!$\cg_{\ckr}$}}
\psfrag{j}{{\tiny \!\!$\cg_{s}$}}
\psfrag{h}{{\tiny \!\!$\cg_{\ckrp}$}}
\psfrag{v}{{\tiny $u$}}
\psfrag{b}{{\tiny $v$}}
\psfrag{c}{{\tiny $0$}}
\psfrag{p}{{\tiny $(\ug(v),v)$}}
\includegraphics[scale=.8]{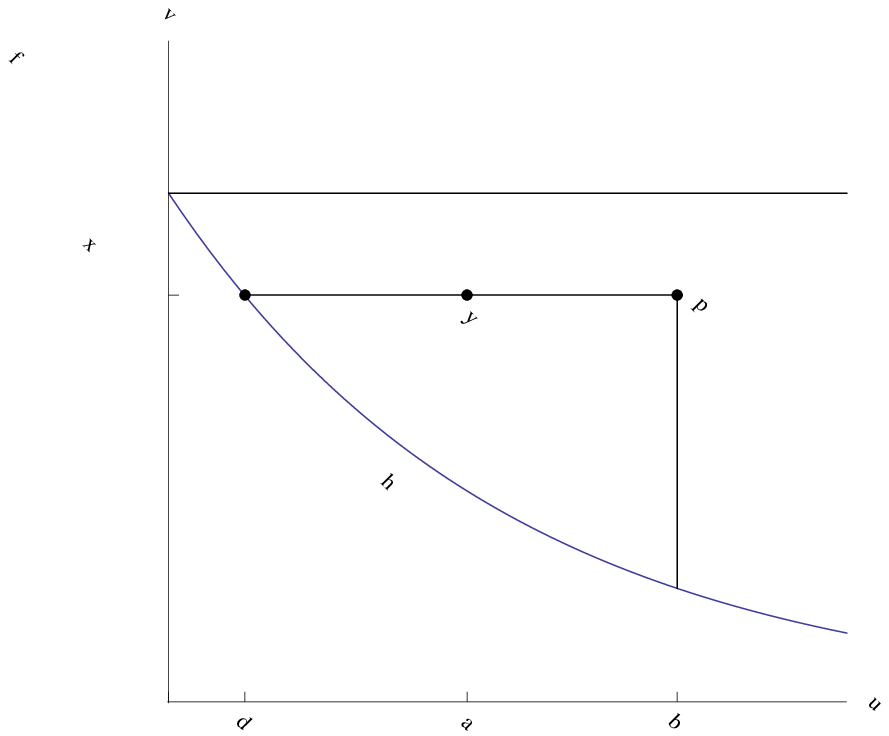}
\end{psfrags}
\end{turn}
\end{center}

\begin{Lem}\label{sto}
Let $\ckrm\leq r\leq\ckrp$. Then 
 \begin{equation}\label{z-0}
\Bigl|\hatz\Bigr|(r,v)\leq \tilde C \max_{u\in[0,\ug(v)]}|\zeta_0|(u),
\end{equation}
 \begin{equation}\label{t-0-a}
\Bigl|\widehat{\frac\theta\lambda}\Bigr|(r,v)\leq C \max_{u\in[0,\ug(v)]}|\zeta_0|(u).
\end{equation}
\end{Lem}
\begin{proof}
From~\eqref{assume} we have
\begin{eqnarray}
 (1-\mu)(r,\varpi)&\leq& (1-\mu)(r,\varpi_0)\nonumber\\ 
 &\leq& \max\left\{(1-\mu)(\ckrm,\varpi_0),(1-\mu)(\ckrp,\varpi_0)\right\}\label{umm}\\
 &=&(1-\mu)(\ckrm,\varpi_0)\nonumber
\end{eqnarray}
and
\begin{eqnarray*}
\frac{\partial_r(1-\mu)}{1-\mu}&=&\frac{2\eta/r^2}{-(1-\mu)}\ \leq\ \frac{2\eta_0/r^2}{-(1-\mu)}\ \leq\  \frac{2\eta_0(\ckrm)/\ckrm^2}{-(1-\mu)}\\
&\leq& \frac{2\eta_0(\ckrm)/\ckrm^2}{-(1-\mu)(r,\varpi_0)}\ \leq\  \frac{2\eta_0(\ckrm)/\ckrm^2}{-(1-\mu)(\ckrm,\varpi_0)}\ =:\ c_{\ckrm}.
\end{eqnarray*}
(For the second inequality, see the graph of $\eta_0$ in Section~\ref{cans}.)
Each of the five exponentials in~\eqref{triple} is bounded by
\begin{equation}\label{c}
e^{c_{\ckrm}(\ckrp-\ckrm)}=:C.
\end{equation}
Hence, for $(r,v)$ such that $(\ug(v),v)\in J^-(\cg_{\ckrm})\cap J^+(\cg_{\ckrp})$, we have from~\eqref{triple}
\begin{eqnarray}
\Bigl|\hatz\Bigr|(r,v)&\leq&C\Bigl|\hatz\Bigr|(\ckrp,\vgrp(\ug(v)))\label{four}\\
&&+C^2\int_{\vgrp(\ug(v))}^{v}\Bigl|\widehat{\frac\theta\lambda}\Bigr|(\ckrp,\tilde v)
\Bigl[\frac{(-\lambda)(\ug(v),\tilde v)}{r(\ug(v),\tilde v)}\Bigr]\,d\tilde v\nonumber\\
&&+C^2\int_{\vgrp(\ug(v))}^{v}\int_{r(\ug(v),\tilde v)}^{\ckrp}\Bigl|\hatz\Bigr|(\tilde s,\tilde v)
\,d\tilde s
\Bigl[\frac{(-\lambda)(\ug(v),\tilde v)}{[r(\ug(v),\tilde v)]^2}\Bigr]
\,d\tilde v.\nonumber
\end{eqnarray}
For $r\leq s\leq \ckrp$, define
\begin{equation}\label{z}
 {\cal Z}_{(r,v)}(s)=\max_{\tilde v\in[\vgs(\ug(v)),v]}\Bigl|\hatz\Bigr|(s,\tilde v)
\end{equation}
and
\begin{equation}\label{t}
{\cal T}_{(r,v)}(\ckrp)=\max_{\tilde v\in[\vgrp(\ug(v)),v]}\Bigl|\widehat{\frac\theta\lambda}\Bigr|(\ckrp,\tilde v).
\end{equation}

\begin{center}
\begin{turn}{45}
\begin{psfrags}
\psfrag{a}{{\tiny $\vgs(\ug(v))$}}
\psfrag{y}{{\tiny $(\ug(v),\vgs(\ug(v)))$}}
\psfrag{u}{{\tiny $v$}}
\psfrag{z}{{\tiny \!\!\!\!$r_+-s_1$}}
\psfrag{x}{{\tiny $\ug(v)$}}
\psfrag{f}{{\tiny $U=r_+-\ckrp$}}
\psfrag{g}{{\tiny \!\!$\cg_{\ckr}$}}
\psfrag{j}{{\tiny \!\!$\cg_{s}$}}
\psfrag{h}{{\tiny \!\!$\cg_{\ckrp}$}}
\psfrag{v}{{\tiny $u$}}
\psfrag{b}{{\tiny $v$}}
\psfrag{c}{{\tiny $0$}}
\psfrag{p}{{\tiny $(\ug(v),v)$}}
\includegraphics[scale=.8]{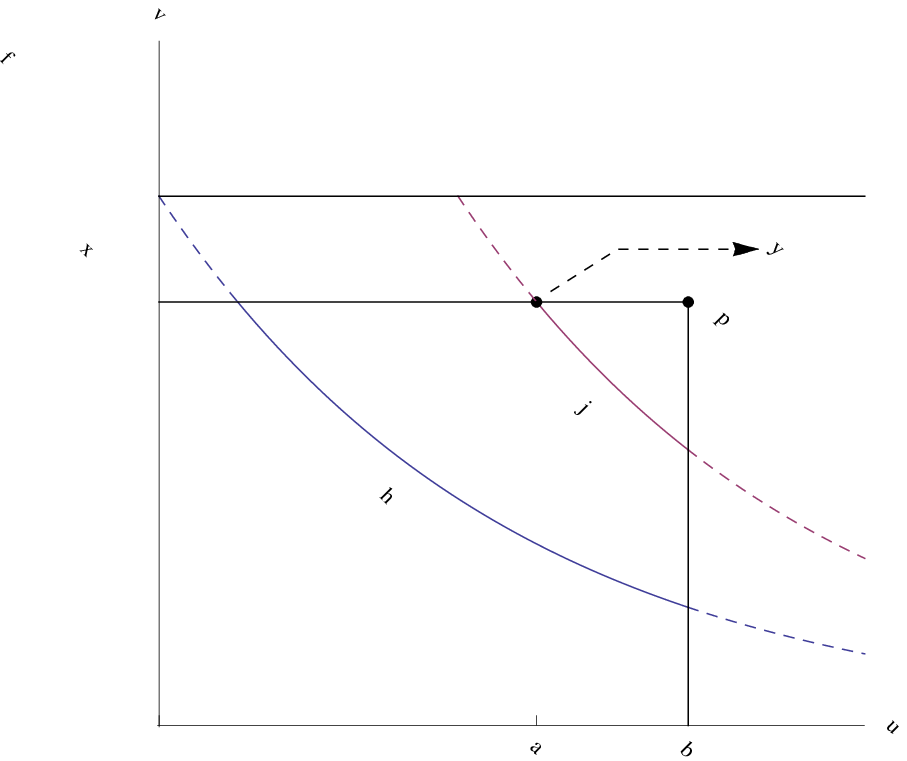}
\end{psfrags}
\end{turn}
\end{center}
\noindent Recall that $[\vgs(\ug(v)),v]$ is the projection of $J^-(\ug(v),v)\cap\cg_s$ on the $v$-axis.

Note that ${\cal Z}_{(r,v)}(r)=\bigl|\hatz\bigr|(r,v)$. Inequality~\eqref{four} implies
\begin{eqnarray}
{\cal Z}_{(r,v)}(r)&\leq& C{\cal Z}_{(r,v)}(\ckrp)\nonumber\\
&&+C^2\int_{\vgrp(\ug(v))}^{v}{\cal T}_{(r,v)}(\ckrp)
\Bigl[\frac{(-\lambda)(\ug(v),\tilde v)}{r(\ug(v),\tilde v)}\Bigr]\,d\tilde v\nonumber\\
&&+C^2\int_{\vgrp(\ug(v))}^{v}\int_{r(\ug(v),\tilde v)}^{\ckrp}{\cal Z}_{(r,v)}(\tilde s)
\,d\tilde s
\Bigl[\frac{(-\lambda)(\ug(v),\tilde v)}{[r(\ug(v),\tilde v)]^2}\Bigr]
\,d\tilde v\nonumber\\
&\leq& C{\cal Z}_{(r,v)}(\ckrp)
+C^2
\ln\Bigl(\frac{\ckrp}{r}\Bigr)
{\cal T}_{(r,v)}(\ckrp)
+C^2
\Bigl(\frac 1r-\frac 1\ckrp\Bigr)
\int_{r}^{\ckrp}{\cal Z}_{(r,v)}(\tilde s)
\,d\tilde s.\nonumber
\end{eqnarray}
Again consider $r\leq s\leq\ckrp$ and let $\tilde v\in [\vgs(\ug(v)),v]$, so that $(\ugs(\tilde v),\tilde v)\in J^-(\ug(v),v)\cap J^+(\cg_{\ckrp})$.
In the same way one can show that
$$
\Bigl|\hatz\Bigr|(s,\tilde v)\leq C{\cal Z}_{(r,v)}(\ckrp)
+C^2
\ln\Bigl(\frac{\ckrp}{s}\Bigr)
{\cal T}_{(r,v)}(\ckrp)
+C^2
\Bigl(\frac 1s-\frac 1\ckrp\Bigr)
\int_{s}^{\ckrp}{\cal Z}_{(r,v)}(\tilde s)
\,d\tilde s,
$$
because $J^-(\ugs(\tilde v),\tilde v)\cap \cg_{\ckrp}\subset J^-(\ug(v),v)\cap \cg_{\ckrp}$.
Therefore,
$$
{\cal Z}_{(r,v)}(s)\leq C{\cal Z}_{(r,v)}(\ckrp)
+C^2
\ln\Bigl(\frac{\ckrp}{r}\Bigr)
{\cal T}_{(r,v)}(\ckrp)
+C^2
\Bigl(\frac 1r-\frac 1\ckrp\Bigr)
\int_{s}^{\ckrp}{\cal Z}_{(r,v)}(\tilde s)
\,d\tilde s.
$$
Using Gronwall's inequality, we get
\begin{equation}\label{zf}
{\cal Z}_{(r,v)}(r)\leq C\Bigl[{\cal Z}_{(r,v)}(\ckrp)+C\ln\Bigl(\frac{\ckrp}{r}\Bigr){\cal T}_{(r,v)}(\ckrp)\Bigr]e^{\frac{C^2(\ckrp-r)^2}{r\ckrp}}.
\end{equation}

To bound ${\cal Z}_{(r,v)}$ and ${\cal T}_{(r,v)}$, it is convenient at this point to use~\eqref{z_00} and
\begin{eqnarray}
\Bigl|\widehat{\frac\theta\lambda}\Bigr|(r,v)&\leq 
&\hat C_r \max_{u\in[0,\ug(v)]}|\zeta_0|(u)\label{t_00}
\end{eqnarray}
(valid for $\ckrp \leq r < r_+$), in spite of having the better estimates~\eqref{z_0_bis} and~\eqref{t_0}.
Indeed, if these better estimates are used, the improvement is just $e^{-\alpha \vgrp(\ug(v))}$ (that is, an exponential factor computed over $\cg_{\ckrp}$ for the same value of $u$); to turn this into an exponential decay in $v$ we must first obtain a more accurate control of the various quantities in the region $\ckrm\leq r\leq\ckrp$.
Applying first the definition~\eqref{z} and then~\eqref{z_00}, we have
\begin{eqnarray}
{\cal Z}_{(r,v)}(\ckrp)&=&\max_{\tilde v\in[\vgrp(\ug(v)),v]}\Bigl|\hatz\Bigr|(\ckrp,\tilde v)\nonumber\\
&\leq&e^{\frac{(r_+-\ckrp)^2}{\ckrp r_+}}\max_{\tilde v\in[\vgrp(\ug(v)),v]}\ \ \max_{u\in[0,\ugrp(\tilde v)]}|\zeta_0|(u)\nonumber\\
&\leq&e^{\frac{(r_+-\ckrp)^2}{\ckrp r_+}}\max_{u\in[0,\ug(v)]}|\zeta_0|(u),\label{part-1a}
\end{eqnarray}
because $\ugrp(\tilde v)\leq \ug(v)$.
 Applying first the definition~\eqref{t} and then~\eqref{t_00}, we have
\begin{eqnarray}
{\cal T}_{(r,v)}(\ckrp)&\leq&e^{\frac{(r_+-\ckrp)^2}{\ckrp r_+}}\ln\Bigl(\frac{r_+}{\ckrp}\Bigr)\max_{\tilde v\in[\vgrp(\ug(v)),v]}\ \
\max_{u\in[0,\ugrp(\tilde v)]}|\zeta_0|(u)\nonumber\\
&\leq&\hat C_{\ckrp}\max_{u\in[0,\ug(v)]}|\zeta_0|(u).\label{part-2a}
\end{eqnarray}
We use~\eqref{part-1a} and~\eqref{part-2a} in~\eqref{zf}. This yields~\eqref{z-0}.

Finally, writing~\eqref{field_6} in the $(r,v)$ coordinates (with $\ckr = \ckrp$) gives
\begin{eqnarray*}
\widehat{\frac\theta\lambda}(r,v)&=&\widehat{\frac\theta\lambda}(\ckrp,v)
e^{\int_r^{\ckrp}\bigl[\frac{1}{\widehat{1-\mu}}\widehat{\partial_r(1-\mu)}\bigr](\tilde{s},v)\,d\tilde{s}}\\
&&+\int_{r}^{\ckrp}\Bigl[\hatz\frac 1 {\tilde s}\Bigr](\tilde s,v)
e^{\int_r^{\tilde s}\bigl[\frac{1}{\widehat{1-\mu}}\widehat{\partial_r(1-\mu)}\bigr](\bar{s},v)\,d\bar{s}}\,d\tilde s.
\end{eqnarray*}
The exponentials are bounded by the constant $C$ in~\eqref{c}. We use the estimates~\eqref{t_00} and~\eqref{z-0} to obtain
\begin{eqnarray*}
\Bigl|\widehat{\frac\theta\lambda}\Bigr|(r,v)&\leq&C\Bigl|\widehat{\frac\theta\lambda}\Bigr|(\ckrp,v)
+C\int_{r}^{\ckrp}\Bigl[\Bigl|\hatz\Bigr|\frac 1 {\tilde s}\Bigr](\tilde s,v)\,d\tilde s\\
&\leq&C\hat C_{\ckrp}\max_{u\in[0,u_{\ckrp}(v)]}|\zeta_0|(u)\\ &&+C
\tilde C\ln\Bigl(\frac{\ckrp}{r}\Bigr) \max_{u\in[0,\ug(v)]}|\zeta_0|(u)\\
&=&C\max_{u\in[0,\ug(v)]}|\zeta_0|(u),
\end{eqnarray*}
which is~\eqref{t-0-a}. 
\end{proof}

According to~\eqref{z_00} and~\eqref{z-0}, the function $\frac\zeta\nu$ is bounded in the region $J^-(\cg_{\ckrm})$, let us say by $\hat\delta$.
Arguing as in the deduction of~\eqref{k_min}, we obtain
\begin{equation}
\kappa(u,v)
\geq\left(\frac{\ckrm}{r_+}\right)^{\hat\delta^2}.\label{k_min_2}
\end{equation}

\begin{Lem}
For $(u,v)\in J^-(\cg_{\ckrm})$, and $U \leq r_+ - \ckrp$ sufficiently small, we have
\begin{equation}\label{mass_2_bis}
\varpi_0\leq\varpi(u,v)\leq \varpi_0+
C\left(\sup_{\tilde u\in[0,u]}|\zeta_0|(\tilde u)\right)^2.
\end{equation}
The curve $\cg_{\ckrm}$ intersects every line of constant $u$. Therefore, $\ugr(\infty)=0$.
\end{Lem}
\begin{proof}
The proof of~\eqref{mass_2_bis} is identical to the proof of~\eqref{mass_2}.

Because $\varpi$ is bounded, the function $1-\mu$ is bounded below in $J^-(\cg_{\ckrm})$.
Also, by~\eqref{umm}, the function $1-\mu$ is bounded above in $J^-(\cg_{\ckrm})\cap J^+(\cg_{\ckrp})$ by $(1-\mu)(\ckrm,\varpi_0)$.

 We claim that for each $0<u\leq U$
\begin{equation}\label{stop2}
\sup\left\{v\in\left[0,\infty\right[:(u,v)\in J^-(\cg_{\ckrm})\right\}<\infty.
\end{equation}
The proof is similar to the proof of~\eqref{stop}: since $\kappa$ is bounded below by a positive constant 
and $1-\mu$ is bounded above by a negative constant, $\lambda$ is bounded above by a negative constant in $J^-(\cg_{\ckrm})\cap J^+(\cg_{\ckrp})$, say $-c_\lambda$.
Then, as long as $(u,v)$ belongs to $J^-(\cg_{\ckrm})$, we have the upper bound for $r(u,v)$ given by
$$
r(u,v)\leq r_+-u-c_\lambda v,
$$
since $0<u\leq U\leq r_+-\ckrp$). Finally, if \eqref{stop2} did not hold for a given $u$, we would have $r(u,v)\to -\infty$
as $v \to \infty$, which is a contradiction. This proves the claim.
\end{proof}

\subsection{Estimates for $\nu$, $\lambda$ and the region $J^-(\cg_{\ckrm})\cap J^+(\cg_{\ckrp})$}\label{nu-lambda}

\begin{Lem}
In the region $J^-(\cg_{\ckrm})\cap J^+(\cg_{\ckrp})$,
we have the following estimates from above and from below on $\lambda$ and $\nu$:
\begin{equation}\label{similar1}
-\tilde C\leq\lambda\leq -\tilde c
\end{equation}
and
\begin{equation}\label{similar2}
-\frac{\tilde C}{u}\leq\nu\leq-\frac{\tilde c}{u},
\end{equation}
where the constants $\tilde c$ and $\tilde C$ depend on $\ckrp$ and $\ckrm$.

Furthermore, if $0 < \delta < r_+-r_0$ and $(u,v)\in\cg_{r_+ - \delta}$ then
\begin{equation}\label{bound_lambda}
-C\,\partial_r(1-\mu)(r_+,\varpi_0)\,\delta\leq\lambda(u,v)\leq -c\,\partial_r(1-\mu)(r_+,\varpi_0)\,\delta
\end{equation}
and
\begin{equation}\label{bound_nu}
-C\frac{\delta}{u}\leq\nu(u,v)\leq-c\frac{\delta}{u},
\end{equation}
where the constants $0<c<1<C$ may be chosen independently of $\delta$. 
Given $\eps>0$ then $1-\eps<c<1$ and $1<C<1+\eps$ for small enough $\delta$.
\end{Lem}
\begin{proof}
From~\eqref{kappa_u} we obtain 
(the Raychaudhuri equation)
$$
\partial_u\left(\frac{\lambda}{1-\mu}\right)=\frac{\lambda}{1-\mu}\left(\frac{\zeta}{\nu}\right)^2\frac{\nu}{r},
$$
and from~\eqref{Ray} we obtain (the Raychaudhuri equation)
\begin{equation}\label{ray_v_bis} 
\partial_v\left(\frac{\nu}{1-\mu}\right)=\frac{\nu}{1-\mu}\left(\frac{\theta}{\lambda}\right)^2\frac{\lambda}{r}.
\end{equation}
Let $\hat\delta>0$. By decreasing $U$, if necessary, using~\eqref{z_0_bis}, \eqref{t_0}, \eqref{z-0} and~\eqref{t-0-a}, we have
$\bigl|\frac{\theta}{\lambda}\bigr|<\hat\delta$ and 
$\bigl|\frac{\zeta}{\nu}\bigr|<\hat\delta$ in $J^-(\cg_{\ckrm})$.
Since 
$\int_0^u\frac\nu r(\tilde u,v)\,d\tilde{u}=\ln\Bigl(\frac{r(u,v)}{r_+}\Bigr)$,
$\int_0^v\frac\lambda r(u,\tilde v)\,d\tilde{v}=\ln\Bigl(\frac{r(u,v)}{r_+-u}\Bigr)$,
$\frac{\ckrm}{r_+}\leq\frac{r(u,v)}{r_+}\leq 1$ and $\frac{\ckrm}{r_+}\leq\frac{r(u,v)}{r_+-u}\leq 1$, 
for $(u,v)\in J^-(\cg_{\ckrm})$ we have
\begin{equation}\label{close_one}
\Bigl(\frac{\ckrm}{r_+}\Bigr)^{\hat\delta^2}\leq e^{\int_0^u\bigl((\frac{\zeta}{\nu})^2\frac{\nu}{r}\bigr)(\tilde u,v)\,d\tilde u}\leq 1,
\end{equation}
\begin{equation}\label{perto_um}
\Bigl(\frac{\ckrm}{r_+}\Bigr)^{\hat\delta^2}\leq e^{\int_0^v\bigl((\frac{\theta}{\lambda})^2\frac{\lambda}{r}\bigr)(u,\tilde v)\,d\tilde v}\leq 1.
\end{equation}
So,
integrating the Raychaudhuri equations, we get
\begin{equation}\label{lambda_mu}
\Bigl(\frac{\ckrm}{r_+}\Bigr)^{\hat\delta^2}=
\Bigl(\frac{\ckrm}{r_+}\Bigr)^{\hat\delta^2}\frac{\lambda}{1-\mu}(0,v)\leq\frac{\lambda}{1-\mu}(u,v)\leq\frac{\lambda}{1-\mu}(0,v)=1
\end{equation}
(as $\kappa(0,v)=1$), and
\begin{equation}\label{nu_mu}
-\Bigl(\frac{\ckrm}{r_+}\Bigr)^{\hat\delta^2}\frac{1}{1-\mu}(u,0)
\leq
\frac{\nu}{1-\mu}(u,v)\leq\frac{\nu}{1-\mu}(u,0)=-\,\frac{1}{1-\mu}(u,0).
\end{equation}
To bound $(1-\mu)(u,0)$, using~\eqref{lambda_u} and~\eqref{kappa_u}, we compute
\begin{equation}\label{mu_u}
 \partial_u(1-\mu)=\partial_u\Bigl(\frac\lambda\kappa\Bigr)
 =\nu\partial_r(1-\mu)-(1-\mu)\frac{\nu}{r}\Bigl(\frac\zeta\nu\Bigr)^2.
\end{equation}
At the point $(u,v)=(0,0)$ this yields
$$
\partial_u(1-\mu)(0,0)=-\partial_r(1-\mu)(r_+,\varpi_0).
$$
Fix $0<\eps<1$. Since the function
$u\mapsto(1-\mu)(u,0)$ is $C^1$, by decreasing $U$ if necessary, we have
$$
-\frac{\partial_r(1-\mu)(r_+,\varpi_0)}{1-\eps}<\partial_u(1-\mu)(u,0)<-\frac{\partial_r(1-\mu)(r_+,\varpi_0)}{1+\eps}
$$
for $0\leq u\leq U$, and so
$$
-\frac{\partial_r(1-\mu)(r_+,\varpi_0)}{1-\eps}\,u<(1-\mu)(u,0)<-\frac{\partial_r(1-\mu)(r_+,\varpi_0)}{1+\eps}\,u.
$$
Using these inequalities in~\eqref{nu_mu} immediately gives
\begin{equation}\label{nu_mu_bis} 
\Bigl(\frac{\ckrm}{r_+}\Bigr)^{\hat\delta^2}
\frac{1-\eps}{\partial_r(1-\mu)(r_+,\varpi_0)\,u}\leq\frac{\nu}{1-\mu}(u,v)\leq\frac{1+\eps}{\partial_r(1-\mu)(r_+,\varpi_0)\,u}.
\end{equation}

To obtain bounds on $\lambda$ and $\nu$ from~\eqref{lambda_mu} and~\eqref{nu_mu_bis},
recall that, in accordance with~\eqref{umm}, in the region $J^-(\cg_{\ckrm})\cap J^+(\cg_{\ckrp})$ the function $1-\mu$ is bounded 
above by a negative constant. 
On the other hand, the bounds we obtained earlier on $\varpi$ in 
$J^-(\cg_{\ckrm})$ imply that $1-\mu$ is bounded below in $J^-(\cg_{\ckrm})$. 
In summary, there exist $\overline{c}$ and $\overline{C}$ such that
$$
-\overline{C}\leq 1-\mu\leq -\overline{c}. 
$$
Therefore, from~\eqref{lambda_mu} and~\eqref{nu_mu_bis}, in the region $J^-(\cg_{\ckrm})\cap J^+(\cg_{\ckrp})$,
we get~\eqref{similar1} and~\eqref{similar2}:
$$
-\overline{C}\leq\lambda\leq-\overline{c}\Bigl(\frac{\ckrm}{r_+}\Bigr)^{\hat\delta^2},
$$
$$
-\overline{C}
\frac{1+\eps}{\partial_r(1-\mu)(r_+,\varpi_0)}\frac 1u
\leq\nu\leq
-\overline{c}\Bigl(\frac{\ckrm}{r_+}\Bigr)^{\hat\delta^2}\frac{1-\eps}{\partial_r(1-\mu)(r_+,\varpi_0)}
\frac1 u.
$$
By decreasing $\overline c$ and increasing $\overline C$, if necessary, we can guarantee~\eqref{similar1} and~\eqref{similar2}
hold, without having to further decrease $U$.

Now suppose that $(u,v)\in\cg_{r_+ - \delta}$. Then
\begin{eqnarray}
(1-\mu)(u,v)&=&(1-\mu)(r_+ - \delta,\varpi)\nonumber\\
&\leq&(1-\mu)(r_+ - \delta,\varpi_0)\nonumber\\
&\leq&-\frac{\partial_r(1-\mu)(r_+,\varpi_0)}{1+\eps}\,\delta,\label{negative}
\end{eqnarray}
where $\eps$ is any fixed positive number, provided that $\delta$ is sufficiently small. If $\delta$ is not small, then~\eqref{negative} also holds
but with $1+\eps$ replaced by a larger constant.

Using again~\eqref{mu_u},
$$
(1-\mu)(u,v)=-\int_0^ue^{-\int_{\tilde u}^u\bigl(\frac{\nu}{r}\bigl(\frac\zeta\nu\bigr)^2\bigr)(\bar u,v)\,d\bar u}
\Bigl(\frac{2\nu}{r^2}\eta\Bigr)(\tilde u,v)\,d\tilde u.
$$
We take into account that
$$
e^{-\int_{\tilde u}^u\bigl(\frac{\nu}{r}\bigl(\frac\zeta\nu\bigr)^2\bigr)(\bar u,v)\,d\bar u}\leq\Bigl(\frac{r_+}{r_+-\delta}\Bigr)^{\hat{\delta}^2}
$$
and
\begin{eqnarray*}
-\,\frac{2\nu}{r^2}\eta&\geq& 2\nu\left(-\,\frac{e^2}{r^3}-\frac\Lambda 3r+\frac{\varpi_0}{r^2}\right)+2\nu\frac{\tilde{\delta}}{r^2}\\
&=&\nu\partial_r(1-\mu)(r,\varpi_0)+2\nu\frac{\tilde{\delta}}{r^2}
\end{eqnarray*}
provided $U$ is chosen small enough so that $\varpi\leq\varpi_0+\tilde{\delta}$ in $J^-(\cg_{r_+-\delta})$. We get
\begin{eqnarray}
(1-\mu)(u,v)
&\geq&\Bigl(\frac{r_+}{r_+-\delta}\Bigr)^{\hat{\delta}^2}
\left((1-\mu)(r_+-\delta,\varpi_0)-\,\frac{2\tilde{\delta}\delta}{r_+(r_+-\delta)}\right)\nonumber\\
&\geq&-\,\Bigl(\frac{r_+}{r_+-\delta}\Bigr)^{\hat{\delta}^2}
\left(\frac{\partial_r(1-\mu)(r_+,\varpi_0)}{1-\eps}+\frac{4\tilde{\delta}}{r_+^2}\right)\,\delta,\label{negative_below}
\end{eqnarray}
where $0<\eps<1$,
provided $\delta$ is sufficiently small. 
We notice that in
the case under consideration the integration is done between $r_+$ and $r_+-\delta$ and so
the left hand sides of~\eqref{close_one} and~\eqref{perto_um} 
can be improved to
$\bigl(\frac{r_+-\delta}{r_+}\bigr)^{\hat\delta^2}$.
Estimates
\eqref{lambda_mu}, \eqref{negative} and~\eqref{negative_below} yield, for $\hat\delta\leq 1$,
\begin{eqnarray} 
-\,\Bigl(\frac{r_+}{r_+-\delta}\Bigr)
\left(\frac{\partial_r(1-\mu)(r_+,\varpi_0)}{1-\eps}+\frac{4\tilde{\delta}}{r_+^2}\right)\,\delta\leq\lambda\qquad\qquad\qquad\qquad&&\label{bound_lambda2}\\
\leq
-\Bigl(\frac{r_+-\delta}{r_+}\Bigr)\frac{\partial_r(1-\mu)(r_+,\varpi_0)}{1+\eps}\,\delta,&&\nonumber
\end{eqnarray}
 whereas estimates
\eqref{nu_mu_bis}, \eqref{negative} and~\eqref{negative_below} yield, again for $\hat\delta\leq 1$,
\begin{eqnarray}
-\,\Bigl(\frac{r_+}{r_+-\delta}\Bigr)
\left(\frac{1+\eps}{1-\eps}+\frac{4\tilde{\delta}(1+\eps)}{r_+^2\partial_r(1-\mu)(r_+,\varpi_0)}\right)\,\frac\delta u\leq\nu
\qquad\qquad\ \ &&\label{bound_nu2}\\
\leq
-\Bigl(\frac{r_+-\delta}{r_+}\Bigr)\frac{1-\eps}{1+\eps}\,\frac\delta u.&&\nonumber
\end{eqnarray}
Estimates~\eqref{bound_lambda}
and~\eqref{bound_nu} are established. Note that $u \leq \delta$
when $(u,v)\in\cg_{r_+-\delta}$. Since
\[
c = c(\delta, \eps, \tilde \delta) = c(\delta, \eps(U,\delta), \tilde \delta(U)) = c(\delta, \eps(U(\delta),\delta), \tilde \delta(U(\delta))),  
\]
and analogously for $C$, we see that $c$ and $C$ can be chosen arbitrarily close to one, provided that $\delta$
is sufficiently small.
\end{proof}

\begin{Lem} Let $\eps>0$. If $\delta$ is sufficiently small, then
for any point $(u,v)\in \cg_{r_+-\delta}$ we have
\begin{equation}\label{region_delta} 
\delta\, e^{-[\partial_r(1-\mu)(r_+,\varpi_0)+\eps]\,v}
\leq u \leq
\delta\, e^{-[\partial_r(1-\mu)(r_+,\varpi_0)-\eps]\,v}.
\end{equation}
For any point $(u,v)\in J^-(\cg_{\ckrm})\cap J^+(\cg_{r_+-\delta})$ we have
\begin{equation}\label{region}
\delta\, e^{-[\partial_r(1-\mu)(r_+,\varpi_0)+\eps]\,v}
\leq u\leq\  \delta\, e^{\frac{r_+}{\tilde c}}e^{-[\partial_r(1-\mu)(r_+,\varpi_0)-\eps]\,v}.
\end{equation}
\end{Lem}
\begin{proof}
Obviously, we have
$$
r(\udz(v),v)=r_+-\delta.
$$
Since $r$ is $C^1$ and $\nu$ does not vanish, $v\mapsto\udz(v)$ is $C^1$.
Differentiating both sides of the last equality with respect to $v$ we obtain
$$
\udp(v)=-\,\frac{\lambda(\udz(v),v)}{\nu(\udz(v),v)}.
$$
Using~\eqref{bound_lambda} and~\eqref{bound_nu}, we have
$$
-\,\frac{C}{c}\partial_r(1-\mu)(r_+,\varpi_0)\udz(v)\leq
\udp(v)\leq-\,\frac{c}{C}\partial_r(1-\mu)(r_+,\varpi_0)\udz(v).
$$
Integrating the last inequalities between $0$ and $v$, as $\udz(0)=\delta$, we have
$$\delta e^{-\,\frac{C}{c}\partial_r(1-\mu)(r_+,\varpi_0)v}\leq
\udz(v)\leq \delta e^{-\,\frac{c}{C}\partial_r(1-\mu)(r_+,\varpi_0)v}.
$$
This proves~\eqref{region_delta}.

Let $(u,v)\in J^-(\cg_{\ckrm})\cap J^+(\cg_{r_+ - \delta})$.
Integrating~\eqref{similar2} between $\udz(v)$ and $u$, we get
$$
1\leq\frac{u}{\udz(v)}\leq e^{\frac{r_+}{\tilde c}}.
$$ 
Combining $\vdz(u)\leq v$ with the first inequality in~\eqref{region_delta} applied at the point $(u,\vdz(u))$,
\begin{eqnarray*}
u&\geq& \delta\, e^{-[\partial_r(1-\mu)(r_+,\varpi_0)+\eps]\,\vdz(u)}\\
&\geq& \delta\, e^{-[\partial_r(1-\mu)(r_+,\varpi_0)+\eps]\,v},
\end{eqnarray*}
and combining $u\leq e^{\frac{r_+}{\tilde c}}\udz(v)$ with the second inequality in~\eqref{region_delta} applied at the point $(\udz(v),v)$,
\begin{eqnarray*}
u&\leq& e^{\frac{r_+}{\tilde c}}\udz(v)\ \leq\  \delta\, e^{\frac{r_+}{\tilde c}}e^{-[\partial_r(1-\mu)(r_+,\varpi_0)-\eps]\,v}.
\end{eqnarray*}
\end{proof}

\subsection{Improved estimates}\label{medium-3}

\begin{Lem}
 Let $\ckrm\leq r\leq\ckrp$. Then 
\begin{equation}\label{z_f}
 \Bigl|\hatz\Bigr|(r,v)\leq \tilde C_{\ckrm} \max_{u\in[0,\ug(v)]}|\zeta_0|(u)e^{-\alpha v}.
\end{equation}
\end{Lem}
\begin{proof}
Applying first the definition~\eqref{z} and then~\eqref{z_0_bis},
\begin{eqnarray}
{\cal Z}_{(r,v)}(\ckrp)&=&\max_{\tilde v\in[\vgrp(\ug(v)),v]}\Bigl|\hatz\Bigr|(\ckrp,\tilde v)\nonumber\\
&\leq&e^{\frac{(r_+-\ckrp)^2}{\ckrp r_+}}\max_{\tilde v\in[\vgrp(\ug(v)),v]}\ \ \max_{u\in[0,\ugrp(\tilde v)]}|\zeta_0|(u)e^{-\alpha \tilde v}\nonumber\\
&\leq&e^{\frac{(r_+-\ckrp)^2}{\ckrp r_+}}\max_{u\in[0,\ugrp(\vgrp(\ug(v)))]}|\zeta_0|(u)e^{-\alpha \vgrp(\ug(v))}\nonumber\\
&=&e^{\frac{(r_+-\ckrp)^2}{\ckrp r_+}}\max_{u\in[0,\ug(v)]}|\zeta_0|(u)e^{-\alpha \vgrp(\ug(v))}\label{zr}
\end{eqnarray}
because $\ugrp(\tilde v)\geq \ugrp(v)$. 
Integrating~\eqref{similar1} between $\vgrp(u_r(v))$ and $v$, we get
\begin{equation}\label{diff}
v-\vgrp(\ug(v))\leq{\textstyle\frac{\ckrp-r}{\tilde c}}=:
 c_{r,\ckrp}\leq c_{\ckrm,\ckrp}.
\end{equation}
This allows us to continue the estimate~\eqref{zr}, to obtain 
\begin{eqnarray}
{\cal Z}_{(r,v)}(\ckrp)\leq e^{\frac{(r_+-\ckrp)^2}{\ckrp r_+}}e^{\alpha c_{\ckrm,\ckrp}}\max_{u\in[0,\ug(v)]}|\zeta_0|(u)e^{-\alpha v}.\label{part-1}
\end{eqnarray}
 Applying first the definition~\eqref{t} and then~\eqref{t_0}, and repeating the computations that lead to~\eqref{zr} and~\eqref{diff},
\begin{eqnarray}
{\cal T}_{(r,v)}(\ckrp)&\leq&e^{\frac{(r_+-\ckrp)^2}{\ckrp r_+}}\ln\Bigl(\frac{r_+}{\ckrp}\Bigr)\max_{\tilde v\in[\vgrp(\ug(v)),v]}\ \
\max_{u\in[0,\ugrp(\tilde v)]}|\zeta_0|(u)e^{-\alpha \tilde v}\nonumber\\
&\leq&e^{\frac{(r_+-\ckrp)^2}{\ckrp r_+}}e^{\alpha c_{\ckrm,\ckrp}}\ln\Bigl(\frac{r_+}{\ckrp}\Bigr)\max_{u\in[0,\ug(v)]}|\zeta_0|(u)e^{-\alpha v}\label{part-2}
\end{eqnarray}
We use~\eqref{part-1} and~\eqref{part-2} in~\eqref{zf}. This yields~\eqref{z_f}.
\end{proof}

\begin{Lem}
 Let $\ckrm\leq r\leq\ckrp$. Then 
 \begin{equation}\label{t-f} 
\Bigl|\widehat{\frac\theta\lambda}\Bigr|(r,v)\leq\overline{C}\max_{u\in[0,\ug(v)]}|\zeta_0|(u)e^{-\alpha v},
\end{equation}
\begin{equation}
|\hat\theta|(r,v)\leq\check C\max_{u\in[0,\ug(v)]}|\zeta_0|(u)e^{-\alpha v}.\label{theta_ff}
\end{equation}
\end{Lem}
\begin{proof}
Just like inequality~\eqref{t-0-a} was obtained from~\eqref{t_00} (that is, \eqref{t_0} with $\alpha=0$) and~\eqref{z-0}, inequality~\eqref{t-f} will be obtained from~\eqref{t_0} and~\eqref{z_f}.
Writing~\eqref{field_6} in the $(r,v)$ coordinates,
\begin{eqnarray*}
\widehat{\frac\theta\lambda}(r,v)&=&\widehat{\frac\theta\lambda}(\ckrp,v)
e^{\int_r^{\ckrp}\bigl[\frac{1}{\widehat{1-\mu}}\widehat{\partial_r(1-\mu)}\bigr](\tilde{s},v)\,d\tilde{s}}\\
&&+\int_{r}^{\ckrp}\Bigl[\hatz\frac 1 {\tilde s}\Bigr](\tilde s,v)
e^{\int_r^{\tilde s}\bigl[\frac{1}{\widehat{1-\mu}}\widehat{\partial_r(1-\mu)}\bigr](\bar{s},v)\,d\bar{s}}\,d\tilde s.
\end{eqnarray*}
The exponentials are bounded by the constant $C$ in~\eqref{c}. We use the estimates~\eqref{t_0} and~\eqref{z_f} to obtain
\begin{eqnarray*}
\Bigl|\widehat{\frac\theta\lambda}\Bigr|(r,v)&\leq&C\Bigl|\widehat{\frac\theta\lambda}\Bigr|(\ckrp,v)
+C\int_{r}^{\ckrp}\Bigl[\Bigl|\hatz\Bigr|\frac 1 {\tilde s}\Bigr](\tilde s,v)\,d\tilde s\\
&\leq&C\hat C_{\ckrp}\max_{u\in[0,\ugrp(v)]}|\zeta_0|(u)e^{-\alpha v}\\ &&+C
\tilde C_{\ckrm}\ln\Bigl(\frac{\ckrp}{r}\Bigr) \max_{u\in[0,\ug(v)]}|\zeta_0|(u)e^{-\alpha v}\\
&=&\overline{C}\max_{u\in[0,\ug(v)]}|\zeta_0|(u)e^{-\alpha v}.
\end{eqnarray*}

Using~\eqref{similar1}, the function $\lambda$ is bounded from below in $J^-(\cg_{\ckrm})\cap J^+(\cg_{\ckrp})$. Hence~\eqref{t-f} implies~\eqref{theta_ff}.
\end{proof}

\begin{Rmk}
For use in\/ {\rm Part~3}, we observe that~\eqref{t_0} and~\eqref{t-f} imply
\begin{equation}\label{canto-part}
\lim_{\stackrel{(u,v)\to(0,\infty)}{{\mbox{\tiny{$(u,v)\in J^-(\Gamma_{\ckrm})$}}}}}\,\Bigl|\frac\theta\lambda\Bigr|(u,v)=0.
\end{equation}
\end{Rmk}

\section{The region $J^-(\gam)\cap J^+(\cg_{\ckrm})$}\label{large}

In this section, we define a curve $\gam$ to the future of $\cg_{\ckrm}$.
Our first aim is to obtain the bounds in Corollary~\ref{cor-8}, $r(u,v)\geq r_--\frac\eps 2$ and $\varpi(u,v)\leq\varpi_0+\frac\eps 2$,
for $(u,v)\in J^-(\gam)\cap J^+(\cg_{\ckrm})$ with $u\leq U_\eps$. In the process, we will bound
$\int_{\ugr(v)}^u\bigl[\bigl|\frac\zeta\nu\bigr||\zeta|\bigr](\tilde u,v)\,d\tilde u$ (this is inequality~\eqref{K}).
Then we will obtain a lower bound on $\kappa$, as well as upper and lower bounds on $\lambda$ and $\nu$. Therefore 
this region, where $r$ may already be below $r_-$, is still a small perturbation of the Reissner-Nordstr\"{o}m solution.

We choose a positive number\footnote{We always have $-\partial_r(1-\mu)(r_-,\varpi_0)>\partial_r(1-\mu)(r_+,\varpi_0)$ (see Appendix~A of Part~3). So,
in particular, we may choose $\beta=-\,\frac{\partial_r(1-\mu)(r_+,\varpi_0)}{\partial_r(1-\mu)(r_-,\varpi_0)}$.}
\begin{equation}\label{beta} 
0<\beta<{\textstyle\frac 12\left(\sqrt{1-8\frac{\partial_r(1-\mu)(r_+,\varpi_0)}{\partial_r(1-\mu)(r_-,\varpi_0)}}-1\right)},
\end{equation}
and define $\gam=\gamma_{\ckrm,\beta}$ to be the curve parametrized by
\begin{equation}\label{small_gamma}
u\mapsto\big(u,(1+\beta)\,\vgr(u)),
\end{equation}
for $u\in[0,U]$.
Since the curve $\cg_{\ckrm}$ is spacelike, so is $\gam$ ($u\mapsto\vgr(u)$ is strictly decreasing).
\begin{Lem}\label{boot_r}
For each $\beta$ satisfying~\eqref{beta} there exist $r_-<\overline{\ckrm}<r_0$ and $0<\eps_0<r_-$ for which, whenever\/ $\ckrm$ and $\eps$ are chosen satisfying $r_-<\ckrm\leq\overline{\ckrm}$ and $0<\eps\leq\eps_0$, the following holds:
there exists $U_\eps$ (depending on $\ckrm$ and $\eps$) such that if $(u,v)\in J^-(\gam)\cap J^+(\cg_{\ckrm})$,
with $0<u\leq U_\eps$, and
\begin{equation}\label{r_omega}
r(u,v)\geq r_--\eps, 
\end{equation}
then 
\begin{equation}\label{r_omega_2}
r(u,v)\geq r_--{\textstyle\frac\eps 2}\quad{\rm and}\quad\varpi(u,v)\leq\varpi_0+{\textstyle\frac\eps 2}.
\end{equation}
\end{Lem}
\begin{Cor}\label{cor-8}
Suppose that $\beta$ is given satisfying~\eqref{beta}, and let $\overline{\ckrm}$ and $\eps_0$ be as in the previous lemma. Fix $r_-<\ckrm\leq\overline{\ckrm}$ and $0<\eps<\eps_0$. If $(u,v)\in J^-(\gam)\cap J^+(\cg_{\ckrm})$ with $0<u\leq U_\eps$, then 
\begin{equation}
r(u,v)\geq r_--{\textstyle\frac\eps 2}\quad{\rm and}\quad\varpi(u,v)\leq\varpi_0+{\textstyle\frac\eps 2}.\label{r-baixo}
\end{equation}
\end{Cor}
\begin{proof}
On $\cg_{\ckrm}$ we have $r=\ckrm>r_->r_--{\textstyle\frac\eps 2}$.
Suppose that there exists a point $(u,v)\in J^-(\gam)\cap 
J^+(\cg_{\ckrm})$, with $0<u\leq U_\eps$, such that $r(u,v)<r_--\frac\eps 2$.
Then there exists a point $(\tilde u,v)$, with $0<\tilde u<u\leq U_\eps$, such that $r_--\eps\leq r(\tilde u,v)<r_--{\frac\eps 2}$.
The point $(\tilde u,v)$ belongs to $J^-(\gam)\cap 
J^+(\cg_{\ckrm})$.
Applying Lemma~\ref{boot_r} at the point $(\tilde u,v)$, we reach a contradiction.
The rest of the argument is immediate.
\end{proof}
\begin{proof}[Proof of\/ {\rm Lemma~\ref{boot_r}}] 
Let $(u,v)\in J^-(\gam)\cap 
J^+(\cg_{\ckrm})$ such that~\eqref{r_omega} holds.
Because of the monotonicity properties of $r$,
$$
\min_{J^-(u,v)\cap 
J^+(\cg_{\ckrm})}r\geq r_--\eps.
$$
According to Proposition~13.2 of~\cite{Dafermos2} (this result depends only on equations \eqref{theta_u} and \eqref{zeta_v}, and so does not depend on the presence of $\Lambda$), there exists a constant $\underline C$ (depending on $r_--\eps_0$) such that 
\begin{eqnarray}
&&\int_{\vgr(u)}^v|\theta|(u,\tilde v)\,d\tilde v+\int_{\ugr(v)}^u|\zeta|(\tilde u,v)\,d\tilde u\label{Big}\\
&&\qquad\leq\underline C\left(
\int_{\vgr(u)}^v|\theta|(\ugr(v),\tilde v)\,d\tilde v+\int_{\ugr(v)}^u|\zeta|(\tilde u,\vgr(u))\,d\tilde u
\right).\nonumber
\end{eqnarray}

\begin{center}
\begin{turn}{45}
\begin{psfrags}
\psfrag{a}{{\tiny $\vgr(u)$}}
\psfrag{y}{{\tiny $(\ug(v),\vgsd(\ug(v)))$}}
\psfrag{u}{{\tiny $v$}}
\psfrag{j}{{\tiny $\cg_{\ckrm}$}}
\psfrag{x}{{\tiny $\ug(v)$}}
\psfrag{f}{{\tiny \!\!\!\!$U$}}
\psfrag{d}{{\tiny \!\!$u$}}
\psfrag{v}{{\tiny $u$}}
\psfrag{b}{{\tiny \!\!\!\!\!\!\!\!\!\!\!\!\!\!\!\!\!$\ugr(v)$}}
\psfrag{c}{{\tiny $v$}}
\psfrag{p}{{\tiny $(u,v)$}}
\includegraphics[scale=.8]{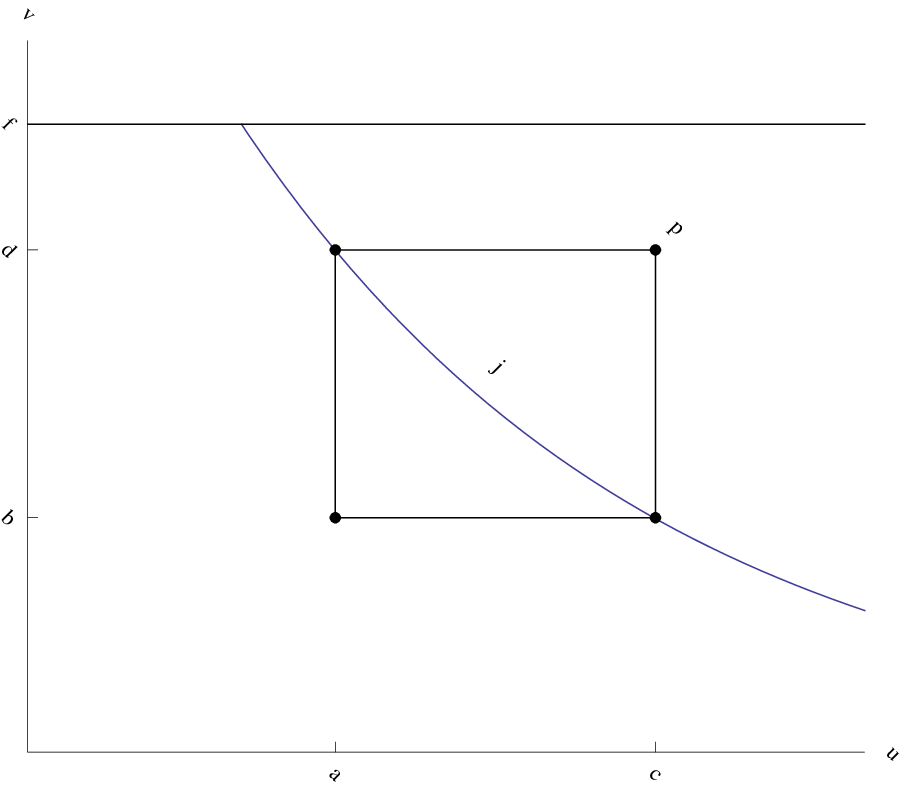}
\end{psfrags}
\end{turn}
\end{center}

The first integral on the right hand side of~\eqref{Big} can be estimated using~\eqref{theta_f}, \eqref{theta_ff} and~\eqref{small_gamma}:
\begin{eqnarray*}
\int_{\vgr(u)}^v|\theta|(\ugr(v),\tilde v)\,d\tilde v
&\leq& C\sup_{[0,u]}|\zeta_0|\int_{\vgr(u)}^ve^{-\alpha\tilde v}\,d\tilde v\\
&\leq&C\sup_{[0,u]}|\zeta_0|e^{-\alpha\vgr(u)}\beta\vgr(u)\\
&\leq&C\sup_{[0,u]}|\zeta_0|e^{-\,\frac\alpha{1+\beta}v}\beta v\\
&\leq&\tilde C\sup_{[0,u]}|\zeta_0|e^{-\,\frac\alpha{1+\beta^+}v},
\end{eqnarray*}
where we have used $\vgr(u)=\frac{\vgam(u)}{1+\beta}\geq\frac{v}{1+\beta}$ and $\vgr(u)\leq v$,
and denoted by $\beta^+$ a fixed number strictly greater than $\beta$.
The second integral on the right hand side of~\eqref{Big} can be estimated using~\eqref{z_0_bis}, \eqref{z_f} and~\eqref{small_gamma}:
\begin{eqnarray*}
\int_{\ugr(v)}^u|\zeta|(\tilde u,\vgr(u))\,d\tilde u&=&\int_{\ugr(v)}^u\Bigl|\frac\zeta\nu\Bigr|(-\nu)(\tilde u,\vgr(u))\,d\tilde u\nonumber\\
&\leq&C(r_+-\ckrm)\sup_{[0,u]}|\zeta_0|e^{-\alpha \vgr(u)}\nonumber\\
&\leq&\tilde C\sup_{[0,u]}|\zeta_0|e^{-\,\frac\alpha{1+\beta}v}.\nonumber
\end{eqnarray*}
These lead to the following estimate for the left hand side of~\eqref{Big}:
\begin{eqnarray}
\int_{\vgr(u)}^v|\theta|(u,\tilde v)\,d\tilde v+\int_{\ugr(v)}^u|\zeta|(\tilde u,v)\,d\tilde u\leq
C\sup_{[0,u]}|\zeta_0|e^{-\,\frac\alpha{1+\beta^+}v}.\label{int_theta_zeta}
\end{eqnarray}
In order to use~\eqref{field_7}, note that, using $\eta(r,\varpi) \leq \eta_0(r)$,
\begin{eqnarray*}
e^{-\int_{\vgr(u)}^{v}\bigl[\frac{\lambda}{1-\mu}\partial_r(1-\mu)\bigr](u,\tilde v)\,d\tilde v}&\leq&
e^{-\partial_r(1-\mu)(r_--\eps_0,\varpi_0)\beta\vgr(u)}\\ &\leq& e^{-\partial_r(1-\mu)(r_--\eps_0,\varpi_0)\beta v}.
\end{eqnarray*}
Thus,
\begin{eqnarray} 
\nonumber
\Bigl|\frac\zeta\nu\Bigr|(u,v)&\leq&\Bigl|\frac\zeta\nu\Bigr|(u,\vgr(u))
e^{-\int_{\vgr(u)}^{v}\bigl[\frac{\lambda}{1-\mu}\partial_r(1-\mu)\bigr](u,\tilde v)\,d\tilde v}
\\
\nonumber
&&+\int_{\vgr(u)}^{v}\frac{|\theta|} r(u,\bar v)e^{-\int_{\bar v}^{v}\bigl[\frac{\lambda}{1-\mu}
\partial_r(1-\mu)\bigr](u,\tilde{v})\,d\tilde{v}}\,d\bar v
\\
\nonumber
&\leq&C\sup_{[0,u]}|\zeta_0|e^{-\alpha \vgr(u)}e^{-\partial_r(1-\mu)(r_--\eps_0,\varpi_0)\beta v}
\\
\nonumber
&&+\frac{e^{-\partial_r(1-\mu)(r_--\eps_0,\varpi_0)\beta v}}{r_--\eps_0}\int_{\vgr(u)}^{v}|\theta|(u,\bar v)\,d\bar v
\\
\nonumber
&\leq&C\sup_{[0,u]}|\zeta_0|e^{-\frac\alpha{1+\beta} v}e^{-\partial_r(1-\mu)(r_--\eps_0,\varpi_0)\beta v}
\\
\nonumber
&&+\frac{e^{-\partial_r(1-\mu)(r_--\eps_0,\varpi_0)\beta v}}{r_--\eps_0}C\sup_{[0,u]}|\zeta_0|e^{-\,\frac\alpha{1+\beta^+}v}
\\
&\leq&C\sup_{[0,u]}|\zeta_0|e^{-\bigl(\frac\alpha{1+\beta^+}+\partial_r(1-\mu)(r_--\eps_0,\varpi_0)\beta\bigr) v}.\label{zetaNuGamma}
\end{eqnarray}
Clearly, the right hand side of the last inequality also bounds $\max_{\tilde u\in[\ugr(v),u]}\bigl|\frac\zeta\nu\bigr|(\tilde u,v)$.
In order to bound $\varpi(u,v)$, note that
\begin{eqnarray}
&&\int_{\ugr(v)}^u\Bigl[\Bigl|\frac\zeta\nu\Bigr||\zeta|\Bigr](\tilde u,v)\,d\tilde u\nonumber\\ 
&&\ \ \ \ \leq
C\sup_{[0,u]}|\zeta_0|e^{-\Bigl(\frac\alpha{1+\beta^+}+\partial_r(1-\mu)(r_--\eps_0,\varpi_0)\beta\Bigr) v}
\int_{\ugr(v)}^u|\zeta|(\tilde u,v)\,d\tilde u\nonumber\\
&&\ \ \ \ \leq C\bigl(\sup_{[0,u]}|\zeta_0|\bigr)^2e^{-\Bigl(\frac{2\alpha}{1+\beta^+}+\partial_r(1-\mu)(r_--\eps_0,\varpi_0)\beta\Bigr) v}.\label{K}
\end{eqnarray}
Using~\eqref{omega_final} and the last estimate, we get
\begin{eqnarray*}
\varpi(u,v)&\leq&
\varpi(\ugr(v),v)e^{\frac{1}{r_--\eps_0}\int_{\ugr(v)}^u\bigl[\bigl|\frac\zeta\nu\bigr||\zeta|\bigr](\tilde u,v)\,d\tilde u}\\
&&+C 
\int_{\ugr(v)}^ue^{\frac{1}{r_--\eps_0}\int_{s}^u\bigl[\bigl|\frac\zeta\nu\bigr||\zeta|\bigr](\tilde u,v)\,d\tilde u}
\Bigl[\Bigl|\frac\zeta\nu\Bigr||\zeta|\Bigr](s,v)\,ds\\
&\leq&\varpi(\ugr(v),v)e^{C\bigl(\sup_{[0,u]}|\zeta_0|\bigr)^2e^{-\Bigl(\frac{2\alpha}{1+\beta^+}+\partial_r(1-\mu)(r_--\eps_0,\varpi_0)\beta\Bigr) v}}\\
&&+Ce^{C\bigl(\sup_{[0,u]}|\zeta_0|\bigr)^2e^{-\Bigl(\frac{2\alpha}{1+\beta^+}+\partial_r(1-\mu)(r_--\eps_0,\varpi_0)\beta\Bigr) v}}\times\\
&&\qquad\qquad\qquad \times
\bigl(\sup_{[0,u]}|\zeta_0|\bigr)^2e^{-\Bigl(\frac{2\alpha}{1+\beta^+}+\partial_r(1-\mu)(r_--\eps_0,\varpi_0)\beta\Bigr) v}.
\end{eqnarray*}
Let $\delta>0$. 
Using the definition of $\alpha$ in~\eqref{alp_new}, the constant in the exponent
\begin{equation}\label{exponent}
\frac{2\alpha}{1+\beta+\delta}+\partial_r(1-\mu)(r_--\eps_0,\varpi_0)\beta
\end{equation}
is positive for 
\begin{equation}\label{success} 
\textstyle
\beta<\frac 12\left(
\sqrt{(1+\delta)^2-8
\frac{\mbox{\tiny $(\frac{\ckrp}{r_+})^{\hat\delta^2}$}\min_{r\in[\ckrp,r_+]}\partial_r(1-\mu)(r,\varpi_0)}{\partial_r(1-\mu)(r_--\eps_0,\varpi_0)}}-(1+\delta)
\right).
\end{equation}
Now, the right hand side tends to
$$
{\textstyle\frac 12\left(\sqrt{1-8\frac{\partial_r(1-\mu)(r_+,\varpi_0)}{\partial_r(1-\mu)(r_-,\varpi_0)}}-1\right)}
$$
as $(\ckrp,\eps_0,\delta)\to(r_+,0,0)$. So, if $\beta$ satisfies~\eqref{beta}, we may choose $\ckrp$, $\eps_0$ and $\delta$ such that~\eqref{success} holds.
Having done this, equations~\eqref{mass_2_bis} and~\eqref{region} now imply that for each $0 <\bar\eps < \eps_0$ there exists $\bar U_{\bar\eps} > 0$ such that
\begin{eqnarray*}
\varpi(u,v)
&\leq&\varpi_0+{\textstyle\frac{\bar\eps} 2},
\end{eqnarray*}
provided that  $u\leq\bar U_{\bar\eps}$.
Since $1-\mu$
is nonpositive and $1-\mu=(1-\mu)(r,\varpi_0)-\frac{2(\varpi-\varpi_0)}{r}$, we have
$$\textstyle
(1-\mu)(r(u,v),\varpi_0)\leq\frac{2(\varpi(u,v)-\varpi_0)}{r}\leq\frac{\bar\eps}{r_--\eps_0}.
$$
Hence, by inspection of the graph of $(1-\mu)(r,\varpi_0)$, there exists $\bar\eps_0$ such that for $0<\bar\eps\leq\bar\eps_0$, we have $r(u,v)>r_--\frac\eps 2$ provided that  $u\leq\bar U_{\bar\eps}$.
For $0<u\leq U_\eps:=\min\{\bar U_{\bar\eps_0},\bar U_\eps\}$, both inequalities~\eqref{r_omega_2} hold.
\end{proof}
\begin{Rmk}
 Given $\eps>0$, we may choose $U$ sufficiently small so that if $(u,v)\in J^-(\gam)\cap J^+(\cg_{\ckrm})$, then
\begin{equation}\label{H} 
 \kappa(u,v)\geq 1-\eps.
\end{equation}
This is a consequence of~\eqref{k_min_2} and~\eqref{K}, since $r$ is bounded away from zero.
\end{Rmk}

Consider the reference subextremal Reissner-Nordstr\"{o}m black hole with renormalized mass $\varpi_0$, charge parameter $e$ and cosmological constant $\Lambda$.
The next remark will turn out to be crucial in Part~3. 
\begin{Rmk}
Suppose that there exist positive constants $C$ and $s$ such that
$|\zeta_0(u)|\leq Cu^s$. Then, instead of choosing $\beta$ according to~\eqref{beta}, in\/ {\rm Lemma~\ref{boot_r}} we may choose
\begin{equation}\label{beta-s} 
0<\beta<{\textstyle\frac 12\left(\sqrt{1-8\frac{(1+s)\partial_r(1-\mu)(r_+,\varpi_0)}{\partial_r(1-\mu)(r_-,\varpi_0)}}-1\right)}.
\end{equation}
\end{Rmk}
\begin{proof}
Let $(u,v)\in J^-(\gam)\cap J^+(\cg_{\ckrm})$.
According to~\eqref{region}, we have
\begin{eqnarray}
u&\leq&Ce^{-[\partial_r(1-\mu)(r_+,\varpi_0)-\eps]\vgr(u)}\nonumber\\
&\leq&Ce^{-[\partial_r(1-\mu)(r_+,\varpi_0)-\eps]\frac v{1+\beta}}\nonumber\\
&\leq&Ce^{-\partial_r(1-\mu)(r_+,\varpi_0)\frac v{1+\beta^+}}.\label{region_bis}
\end{eqnarray}
Thus, the exponent in the upper bound for $\varpi$ in~\eqref{exponent} may be replaced by
$$
\frac{2s\partial_r(1-\mu)(r_+,\varpi_0)}{1+\beta+\delta}+
\frac{2\alpha}{1+\beta+\delta}+\partial_r(1-\mu)(r_--\eps_0,\varpi_0)\beta.
$$
This is positive for
\begin{equation}\label{success-s} 
\textstyle
\beta<\frac 12\left(
\sqrt{(1+\delta)^2-8
\frac{\mbox{\tiny $[(\frac{\ckrp}{r_+})^{\hat\delta^2}$}+s]\min_{r\in[\ckrp,r_+]}\partial_r(1-\mu)(r,\varpi_0)}{\partial_r(1-\mu)(r_--\eps_0,\varpi_0)}}-(1+\delta)
\right).
\end{equation}
Given $\beta$ satisfying~\eqref{beta-s}, we can guarantee that it satisfies the condition above by choosing $(\ckrp,\eps_0,\delta)$ sufficiently close to $(r_+,0,0)$.
\end{proof}
\begin{Cor}
If, for example, $|\zeta_0|(u)\leq e^{-1/{u^2}}$, then instead of choosing $\beta$ according to~\eqref{beta}, in\/ {\rm Lemma~\ref{boot_r}}
we may choose any positive $\beta$.
\end{Cor}

\begin{Lem}\label{lambda-and-nu}
Suppose that $\beta$ is given satisfying~\eqref{beta}.
Choose $\ckrm$ and $\eps_0$ as in the statement of\/ {\rm Lemma~\ref{boot_r}}.
Let $\gam$ be the curve parametrized by~\eqref{small_gamma}.
Let also $\delta>0$, $\beta^-<\beta$ and $\beta^+>\beta$.
There exist constants, $\tilde c$, $\tilde C$, $\overline c$ and $\overline C$, such that
 for  $(u,v)\in\gam$, with $0<u\leq U_{\eps_0}$, we have
 \begin{eqnarray} 
 &&
 \tilde ce^{(1+\delta)\partial_r(1-\mu)(r_--\eps_0,\varpi_0)\frac\beta{1+\beta}\,v}\label{lambda-below}\\
 &&\qquad\qquad\qquad\qquad\leq
  -\lambda(u,v)\leq\nonumber\\
  &&\qquad\qquad\qquad\qquad\qquad\qquad
  \tilde Ce^{(1-\delta)\partial_r(1-\mu)(\ckrm,\varpi_0)\frac\beta{1+\beta}\,v}\label{lambda-above}
 \end{eqnarray}
 and
  \begin{eqnarray}
 &&
 \overline 
 cu^{\mbox{\tiny$-\,\frac{1+\beta^+}{1+\beta^-}\frac{\partial_r(1-\mu)(r_--\eps_0,\varpi_0)}{\partial_r(1-\mu)(r_+,\varpi_0)}\beta$}\,-1}
 \nonumber\\ 
 &&\qquad\qquad\qquad\qquad\leq
  -\nu(u,v)\leq\nonumber\\
  &&\qquad\qquad\qquad\qquad\qquad\qquad
  \overline Cu^{\mbox{\tiny$-\,\frac{1+\beta^-}{1+\beta^+}\frac{\partial_r(1-\mu)(\ckrm,\varpi_0)}{\partial_r(1-\mu)(r_+,\varpi_0)}\beta$}\,-1}.
  \label{nu-above}
  \end{eqnarray}
\end{Lem}
\begin{proof}
Let us first outline the proof.
According to~\eqref{lambda_u} and~\eqref{nu_v},
\begin{eqnarray} 
-\lambda(u,v)&=&-\lambda(\ugr(v),v)e^{\int_{\ugr(v)}^{u}\bigl[\frac{\nu}{1-\mu}\partial_r(1-\mu)\bigr](\tilde u,v)\,d\tilde u},\label{l}\\
-\nu(u,v)&=&-\nu(u,\vgr(u))e^{\int_{\vgr(u)}^v\bigl[\kappa\partial_r(1-\mu)\bigr](u,\tilde v)\,d\tilde v}.\label{n}
\end{eqnarray}
In this region we cannot proceed as was done in the previous section because we cannot guarantee $1-\mu$ is bounded away from zero.
The idea now is to use these two equations to estimate $\lambda$ and $\nu$. For this we need to obtain lower and upper bounds for
\begin{equation}\label{int-1}
\int_{\ugr(v)}^u\frac\nu{1-\mu}(\tilde u,v)\,d\tilde u
\end{equation} 
and 
\begin{equation}\label{int-2}
\int_{\vgr(u)}^v\kappa(u,\tilde v)\,d\tilde v,
\end{equation}
when $(u,v)\in J^-(\gam)\cap J^+(\cg_{\ckrm})$.
The estimates for~\eqref{int-2}, and thus for $\nu$, are easy to obtain.
We estimate~\eqref{int-1} by comparing it with 
\begin{equation}\label{int-3}
\int_{\ugr(v)}^u\frac\nu{1-\mu}(\tilde u,\vgr(\tilde u))\,d\tilde u.
\end{equation}
Using~\eqref{ray_v_bis}, we see that~\eqref{int-1} is bounded above by~\eqref{int-3}. We can also bound \eqref{int-1} from below by \eqref{int-3}, divided by $1+\eps$, once we show that
$$
e^{\int_{\vgr(\bar u)}^v\bigl[\big|\frac\theta\lambda\bigr|\frac{|\theta|}{r}\bigr](\bar u,\tilde v)\,d\tilde v}\leq 1+\eps.
$$
The estimates for $\frac\theta\lambda$ are obtained via~\eqref{field_6} and via upper estimates for~\eqref{int-1}.
To bound~\eqref{int-3} we use the fact that the integrals of $\nu$ and $\lambda$ along $\cg_{\ckrm}$ coincide.

We start the proof by differentiating the equation
$$
r(\tilde u,\vgr(\tilde u))={\ckrm}
$$
with respect to $\tilde u$, obtaining
\begin{equation}
\nu(\tilde u,\vgr(\tilde u))+\lambda(\tilde u,\vgr(\tilde u))\vgr'(\tilde u)=0.\label{derivada}
\end{equation}
For $(u,v)\in J^+(\cg_{\ckrm})$,
integrating~\eqref{derivada} between $\ugr(v)$ and $u$, we get
$$
\int_{\ugr(v)}^u\nu(\tilde u,\vgr(\tilde u))\,d\tilde u+\int_{\ugr(v)}^u\lambda(\tilde u,\vgr(\tilde u))\vgr'(\tilde u)\,d\tilde u =0.
$$
By making the change of variables $\tilde v=\vgr(\tilde u)$, this last equation can be rewritten as
\begin{equation} 
 \int_{\ugr(v)}^u\nu(\tilde u,\vgr(\tilde u))\,d\tilde u-\int_{\vgr(u)}^v\lambda(\ugr(\tilde v),\tilde v)\,d\tilde v=0,
\label{integral}
\end{equation}
as $\vgr(\ugr(v))=v$ and $\frac{d\tilde v}{d\tilde u}=\vgr'(\tilde u)$.

We may bound the integral of $\lambda$ along $\cg_{\ckrm}$ in terms of the integral of $\kappa$ along $\cg_{\ckrm}$ in the following way:
\begin{eqnarray} 
&&-\max_{\cg_{\ckrm}}(1-\mu)
\int_{\vgr(u)}^v\kappa(\ugr(\tilde v),\tilde v)\,d\tilde v\label{F}\\
&&\qquad\qquad\leq -\int_{\vgr(u)}^v\lambda(\ugr(\tilde v),\tilde v)\,d\tilde v\leq\nonumber \\
&&\qquad\qquad\qquad\qquad-\min_{\cg_{\ckrm}}(1-\mu)
\int_{\vgr(u)}^v\kappa(\ugr(\tilde v),\tilde v)\,d\tilde v.\label{A}
\end{eqnarray}
Analogously, we may bound the integral of $\nu$ along $\cg_{\ckrm}$ in terms of the integral of $\frac{\nu}{1-\mu}$ along $\cg_{\ckrm}$ in the following way:
\begin{eqnarray} 
&&-\max_{\cg_{\ckrm}}(1-\mu)
\int_{\ugr(v)}^u\frac\nu{1-\mu}(\tilde u,\vgr(\tilde u))\,d\tilde u\label{B}\\
&&\qquad\ \ \leq-\int_{\ugr(v)}^u\nu(\tilde u,\vgr(\tilde u))\,d\tilde u\leq\nonumber \\
&&\qquad\qquad\qquad -\min_{\cg_{\ckrm}}(1-\mu)
\int_{\ugr(v)}^u\frac\nu{1-\mu}(\tilde u,\vgr(\tilde u))\,d\tilde u.\label{D}
\end{eqnarray}

Let now $(u,v)\in J^-(\gam)\cap J^+(\cg_{\ckrm})$. Using successively~\eqref{ray_v_bis},
\eqref{B}, \eqref{integral} and~\eqref{A}, we get
\begin{eqnarray} 
&&\int_{\ugr(v)}^u\frac\nu{1-\mu}(\tilde u,v)\,d\tilde u\nonumber\\
&&\qquad\qquad\leq\int_{\ugr(v)}^u\frac\nu{1-\mu}(\tilde u,\vgr(\tilde u))\,d\tilde u\nonumber\\
&&\qquad\qquad\leq \mbox{\tiny$\frac{1}{-\max_{\cg_{\ckrm}}(1-\mu)}$}\int_{\ugr(v)}^u-\nu(\tilde u,\vgr(\tilde u))\,d\tilde u\nonumber\\
&&\qquad\qquad= \mbox{\tiny$\frac{1}{-\max_{\cg_{\ckrm}}(1-\mu)}$}\int_{\vgr(u)}^v-\lambda(\ugr(\tilde v),\tilde v)\,d\tilde v\nonumber\\
&&\qquad\qquad\leq \mbox{\tiny$\frac{\min_{\cg_{\ckrm}}(1-\mu)}{\max_{\cg_{\ckrm}}(1-\mu)}$}\int_{\vgr(u)}^v\kappa(\ugr(\tilde v),\tilde v)\,d\tilde v\label{argue3}\\
&&\qquad\qquad\leq \mbox{\tiny$\frac{\min_{\cg_{\ckrm}}(1-\mu)}{\max_{\cg_{\ckrm}}(1-\mu)}$}\beta \vgr(u)\label{G}\\
&&\qquad\qquad \leq \mbox{\tiny$\frac{\min_{\cg_{\ckrm}}(1-\mu)}{\max_{\cg_{\ckrm}}(1-\mu)}$}\,\beta v.\nonumber
\end{eqnarray}

We can now bound the field $\frac\theta\lambda$ for $(u,v)\in J^-(\gam)\cap J^+(\cg_{\ckrm})$.
Using~\eqref{field_6},
\begin{eqnarray}
\Bigl|\frac\theta\lambda\Bigr|(u,v)&\leq &\Bigl|\frac\theta\lambda\Bigr|(\ugr(v),v)
e^{-\int_{\ugr(v)}^{u}\bigl[\frac{\nu}{1-\mu}\partial_r(1-\mu)\bigr](\tilde u,v)\,d\tilde u}\nonumber\\
&&+\int_{\ugr(v)}^{u}\frac{|\zeta|} r(\bar u,v)e^{-\int_{\bar u}^{u}\bigl[\frac{\nu}{1-\mu}\partial_r(1-\mu)\bigr](\tilde{u},v)\,d\tilde{u}}\,d\bar u.\label{C}
\end{eqnarray}
We can bound the exponentials in~\eqref{C} by
\begin{eqnarray*}
&&e^{-\int_{\bar u}^{u}\bigl[\frac{\nu}{1-\mu}\partial_r(1-\mu)\bigr](\tilde u,v)\,d\tilde u}\\
&&\qquad\qquad\leq e^{-\partial_r(1-\mu)(r_--\eps_0,\varpi_0)\int_{\bar u}^{u}\bigl[\frac{\nu}{1-\mu}\bigr](\tilde u,v)\,d\tilde u}\\
&&\qquad\qquad\leq e^{-\partial_r(1-\mu)(r_--\eps_0,\varpi_0)\int_{\ugr(v)}^{u}\bigl[\frac{\nu}{1-\mu}\bigr](\tilde u,v)\,d\tilde u}\\
&&\qquad\qquad\leq e^{-\partial_r(1-\mu)(r_--\eps_0,\varpi_0)\mbox{\tiny$\frac{\min_{\cg_{\ckrm}}(1-\mu)}{\max_{\cg_{\ckrm}}(1-\mu)}$}\beta v}.
\end{eqnarray*}
Combining this inequality with~\eqref{t-f}, \eqref{r-baixo} and~\eqref{int_theta_zeta}, 
leads to
\begin{eqnarray} 
\nonumber
&&\Bigl|\frac\theta\lambda\Bigr|(u,v)
\leq \left(C\sup_{[0,u]}|\zeta_0|e^{-\alpha v}+C\sup_{[0,u]}|\zeta_0|e^{-\frac\alpha{1+\beta^+} v}\right)\times
\\
\nonumber
&&\qquad\qquad\qquad\qquad\qquad\qquad\times
e^{-\partial_r(1-\mu)(r_--\eps_0,\varpi_0)\mbox{\tiny$\frac{\min_{\cg_{\ckrm}}(1-\mu)}{\max_{\cg_{\ckrm}}(1-\mu)}$}\beta v}\\
&&\qquad\qquad\leq C\sup_{[0,u]}|\zeta_0|e^{-\left(\frac\alpha{1+\beta^+}+
\partial_r(1-\mu)(r_--\eps_0,\varpi_0)\mbox{\tiny$\frac{\min_{\cg_{\ckrm}}(1-\mu)}{\max_{\cg_{\ckrm}}(1-\mu)}$}\beta\right)v}.\label{thetaLambdaGamma}
\end{eqnarray}
We consider the two possible cases. Suppose first that the exponent in~\eqref{thetaLambdaGamma} is nonpositive.
Then, from~\eqref{int_theta_zeta} we get
\begin{eqnarray*}
\int_{\vgr(u)}^v\Bigl[\Big|\frac\theta\lambda\Bigr||\theta|\Bigr](u,\tilde v)\,d\tilde v
&\leq& C \int_{\vgr(u)}^v|\theta|(u,\tilde v)\,d\tilde v
\\ &\leq&
C\sup_{[0,u]}|\zeta_0|e^{-\,\frac{\alpha}{1+\beta^+}v}.
\end{eqnarray*}
Suppose now the exponent in~\eqref{thetaLambdaGamma} is positive.
Using~\eqref{thetaLambdaGamma} and~\eqref{int_theta_zeta} again,
\begin{eqnarray*}
&&\int_{\vgr(u)}^v\Bigl[\Big|\frac\theta\lambda\Bigr||\theta|\Bigr](u,\tilde v)\,d\tilde v\\ 
&&\ \leq C\sup_{[0,u]}|\zeta_0|e^{-\left(\frac\alpha{1+\beta^+}+
\partial_r(1-\mu)(r_--\eps_0,\varpi_0)\mbox{\tiny$\frac{\min_{\cg_{\ckrm}}(1-\mu)}{\max_{\cg_{\ckrm}}(1-\mu)}$}\beta\right)v}\int_{\vgr(u)}^v|\theta|(u,\tilde v)\,d\tilde v
\\
&&\ \leq
C(\sup_{[0,u]}|\zeta_0|)^2e^{-\left(\frac{2\alpha}{1+\beta^+}+
\partial_r(1-\mu)(r_--\eps_0,\varpi_0)\mbox{\tiny$\frac{\min_{\cg_{\ckrm}}(1-\mu)}{\max_{\cg_{\ckrm}}(1-\mu)}$}\beta\right)v}.
\end{eqnarray*}
Therefore, in either case, given $\eps>0$ we may choose $U$ sufficiently small so that if $(u,v)\in J^-(\gam)\cap J^+(\cg_{\ckrm})$, then
\begin{equation}\label{integrating_factor} 
e^{\frac 1{r_--\eps_0}\int_{\vgr(\bar u)}^v\bigl[\big|\frac\theta\lambda\bigr||\theta|\bigr](\bar u,\tilde v)\,d\tilde v}\leq 1+\eps,
\end{equation}
for $\bar u\in[\ugr(v),u]$. 

Next we use~\eqref{ray_v_bis}, \eqref{B}, \eqref{D} and~\eqref{integrating_factor}.
We may bound the integral of $\nu$ along $\cg_{\ckrm}$ in terms of the integral of $\frac{\nu}{1-\mu}$ on the segment
$\bigl[\ugr(v),u\bigr]\times\{v\}$ in the following way:
\begin{eqnarray}
&&-\max_{\cg_{\ckrm}}(1-\mu)
\int_{\ugr(v)}^u\frac\nu{1-\mu}(\tilde u,v)\,d\tilde u\label{Q}\\
&&\leq-\max_{\cg_{\ckrm}}(1-\mu)
\int_{\ugr(v)}^u\frac\nu{1-\mu}(\tilde u,\vgr(\tilde u))\,d\tilde u\nonumber\\
&&\qquad\ \ \leq-\int_{\ugr(v)}^u\nu(\tilde u,\vgr(\tilde u))\,d\tilde u\leq\nonumber \\
&&\qquad\qquad\qquad-\min_{\cg_{\ckrm}}(1-\mu)
\int_{\ugr(v)}^u\frac\nu{1-\mu}(\tilde u,\vgr(\tilde u))\,d\tilde u\leq\nonumber\\
&&\qquad\qquad\qquad-(1+\eps)\min_{\cg_{\ckrm}}(1-\mu)
\int_{\ugr(v)}^u\frac\nu{1-\mu}(\tilde u,v)\,d\tilde u.\label{E}
\end{eqnarray}

Now we consider $(u,v)\in\gam$. 
In~\eqref{G} we obtained an upper bound for $\int_{\ugr(v)}^u\frac\nu{1-\mu}(\tilde u,v)\,d\tilde u$. Now we use~\eqref{E} to obtain a lower bound for this quantity.
Applying successively~\eqref{E}, \eqref{integral}, \eqref{F}, and~\eqref{H},
\begin{eqnarray} 
&&\int_{\ugr(v)}^u\frac\nu{1-\mu}(\tilde u,v)\,d\tilde u\nonumber\\
&&\qquad\qquad\geq \mbox{\tiny$\frac{1}{-(1+\eps)\min_{\cg_{\ckrm}}(1-\mu)}$}\int_{\ugr(v)}^u-\nu(\tilde u,\vgr(\tilde u))\,d\tilde u\nonumber\\
&&\qquad\qquad= \mbox{\tiny$\frac{1}{-(1+\eps)\min_{\cg_{\ckrm}}(1-\mu)}$}\int_{\vgr(u)}^v-\lambda(\ugr(\tilde v),\tilde v)\,d\tilde v\nonumber\\
&&\qquad\qquad\geq \mbox{\tiny$\frac{\max_{\cg_{\ckrm}}(1-\mu)}{(1+\eps)\min_{\cg_{\ckrm}}(1-\mu)}$}\int_{\vgr(u)}^v\kappa(\ugr(\tilde v),\tilde v)\,d\tilde v\label{argue}\\
&&\qquad\qquad\geq \mbox{\tiny$\frac{(1-\eps)\max_{\cg_{\ckrm}}(1-\mu)}{(1+\eps)\min_{\cg_{\ckrm}}(1-\mu)}$}\beta \vgr(u)\nonumber\\
&&\qquad\qquad = \mbox{\tiny$\frac{(1-\eps)}{(1+\eps)}\frac{\max_{\cg_{\ckrm}}(1-\mu)}{\min_{\cg_{\ckrm}}(1-\mu)}$}\,{\textstyle\frac\beta{1+\beta}} v.\nonumber
\end{eqnarray}
Thus,
\begin{eqnarray} 
&&e^{\int_{\ugr(v)}^{u}\bigl[\frac{\nu}{1-\mu}\partial_r(1-\mu)\bigr](\tilde u,v)\,d\tilde u}\nonumber\\
&&\qquad\qquad\leq
e^{\bigl[\max_{J^-(\gam)\cap J^+(\cg_{\ckrm})}\partial_r(1-\mu)\bigr]
\int_{\ugr(v)}^{u}\frac\nu{1-\mu}(\tilde u,v)\,d\tilde u}\nonumber\\
&&\qquad\qquad\leq
e^{\bigl[\max_{J^-(\gam)\cap J^+(\cg_{\ckrm})}\partial_r(1-\mu)\bigr]
\mbox{\tiny$\frac{(1-\eps)}{(1+\eps)}\frac{\max_{\cg_{\ckrm}}(1-\mu)}{\min_{\cg_{\ckrm}}(1-\mu)}$}\, \frac\beta{1+\beta} v}\nonumber\\
&&\qquad\qquad\leq
e^{\bigl[\partial_r(1-\mu)(\ckrm,\varpi_0)+\max_{J^-(\gam)\cap J^+(\cg_{\ckrm})}\mbox{\tiny$\frac{2(\varpi-\varpi_0)}{r^2}$}\bigr]
\mbox{\tiny$\frac{(1-\eps)}{(1+\eps)}\frac{\max_{\cg_{\ckrm}}(1-\mu)}{\min_{\cg_{\ckrm}}(1-\mu)}$}\,\frac\beta{1+\beta} v}\nonumber\\
&&\qquad\qquad\leq
e^{\bigl[\partial_r(1-\mu)(\ckrm,\varpi_0)+\mbox{\tiny$\frac\eps{(r_--\eps_0)^2}$}\bigr]
\mbox{\tiny$\frac{(1-\eps)}{(1+\eps)}\frac{\max_{\cg_{\ckrm}}(1-\mu)}{\min_{\cg_{\ckrm}}(1-\mu)}$}\, \frac{\beta}{1+\beta} v}.\label{I}
\end{eqnarray}
On the other hand, using~\eqref{G},
\begin{eqnarray}
&&e^{\int_{\ugr(v)}^{u}\bigl[\frac{\nu}{1-\mu}\partial_r(1-\mu)\bigr](\tilde u,v)\,d\tilde u}\nonumber\\
&&\qquad\qquad\geq
e^{\bigl[\min_{J^-(\gam)\cap J^+(\cg_{\ckrm})}\partial_r(1-\mu)\bigr]
\int_{\ugr(v)}^{u}\frac\nu{1-\mu}(\tilde u,v)\,d\tilde u}\nonumber\\
&&\qquad\qquad\geq
e^{\bigl[\min_{J^-(\gam)\cap J^+(\cg_{\ckrm})}\partial_r(1-\mu)\bigr]
\mbox{\tiny$\frac{\min_{\cg_{\ckrm}}(1-\mu)}{\max_{\cg_{\ckrm}}(1-\mu)}$}\, \frac\beta{1+\beta}v}\nonumber\\
&&\qquad\qquad\geq
e^{\partial_r(1-\mu)(r_--\eps_0,\varpi_0)
\mbox{\tiny$\frac{\min_{\cg_{\ckrm}}(1-\mu)}{\max_{\cg_{\ckrm}}(1-\mu)}$}\, \frac\beta{1+\beta}v}.\label{J}
\end{eqnarray}

We continue assuming $(u,v)\in\gam$. Taking into account~\eqref{similar1},
estimate \eqref{I} allows us to obtain an upper bound for $-\lambda(u,v)$,
\begin{eqnarray*}
-\lambda(u,v)&=&-\lambda(\ugr(v),v)e^{\int_{\ugr(v)}^{u}\bigl[\frac{\nu}{1-\mu}\partial_r(1-\mu)\bigr](\tilde u,v)\,d\tilde u}\\
&\leq& C
e^{\mbox{\tiny$\frac{(1-\eps)}{(1+\eps)}\frac{\max_{\cg_{\ckrm}}(1-\mu)}{\min_{\cg_{\ckrm}}(1-\mu)}$}
\bigl[\partial_r(1-\mu)(\ckrm,\varpi_0)+\mbox{\tiny$\frac\eps{(r_--\eps_0)^2}$}\bigr]
\,\frac\beta{1+\beta} v}\\
&\leq&\tilde Ce^{(1-\delta)\partial_r(1-\mu)(\ckrm,\varpi_0)\frac\beta{1+\beta}\,v},
\end{eqnarray*}
and estimate~\eqref{J} allows us to obtain a lower bound for $-\lambda(u,v)$,
\begin{eqnarray*}
-\lambda(u,v)&=&-\lambda(\ugr(v),v)e^{\int_{\ugr(v)}^{u}\bigl[\frac{\nu}{1-\mu}\partial_r(1-\mu)\bigr](\tilde u,v)\,d\tilde u}\\
&\geq &c
e^{\mbox{\tiny$\frac{\min_{\cg_{\ckrm}}(1-\mu)}{\max_{\cg_{\ckrm}}(1-\mu)}\partial_r(1-\mu)(r_--\eps_0,\varpi_0)\frac\beta{1+\beta}$}
\, v}\\
&\geq&\tilde ce^{(1+\delta)\partial_r(1-\mu)(r_--\eps_0,\varpi_0)\frac\beta{1+\beta}\,v}.
\end{eqnarray*}
Next, we turn to the estimates on $\nu$. Let, again, $(u,v)\in \gam$. Using~\eqref{H},
\begin{eqnarray*}
(1-\eps){\textstyle\frac\beta{1+\beta}}\,v\leq\int_{\vgr(u)}^v\kappa(u,\tilde v)\,d\tilde v\leq{\textstyle\frac\beta{1+\beta}}\,v.
\end{eqnarray*}
These two inequalities imply
\begin{eqnarray}
&&e^{\int_{\vgr(u)}^v[\kappa\partial_r(1-\mu)](u,\tilde v)\,d\tilde v}\nonumber\\
&&\qquad\qquad\leq
e^{\bigl[\max_{J^-(\gam)\cap J^+(\cg_{\ckrm})}\partial_r(1-\mu)\bigr]
\int_{\vgr(u)}^v\kappa(u,\tilde v)\,d\tilde v}\nonumber\\
&&\qquad\qquad\leq
e^{(1-\eps)\bigl[\partial_r(1-\mu)(\ckrm,\varpi_0)+\mbox{\tiny$\frac\eps{(r_--\eps_0)^2}$}\bigr]
\mbox{\tiny$\frac\beta{1+\beta}$}\, v}\label{L}
\end{eqnarray}
and
\begin{eqnarray}
&&e^{\int_{\vgr(u)}^v[\kappa\partial_r(1-\mu)](u,\tilde v)\,d\tilde v}\nonumber\\
&&\qquad\qquad\geq
e^{\bigl[\min_{J^-(\gam)\cap J^+(\cg_{\ckrm})}\partial_r(1-\mu)\bigr]
\int_{\vgr(u)}^v\kappa(u,\tilde v)\,d\tilde v}\nonumber\\
&&\qquad\qquad\geq
e^{\partial_r(1-\mu)(r_--\eps_0,\varpi_0)
\mbox{\tiny$\frac\beta{1+\beta}$}\, v}.\label{M}
\end{eqnarray}
We note that according to~\eqref{region} we have 
\begin{eqnarray} 
&&ce^{-\partial_r(1-\mu)(r_+,\varpi_0)\frac{v}{1+\beta^-}}\label{O}\\
&&\leq ce^{-[\partial_r(1-\mu)(r_+,\varpi_0)+\tilde\eps]\frac{v}{1+\beta}}\nonumber\\
&&= ce^{-[\partial_r(1-\mu)(r_+,\varpi_0)+\tilde\eps]\vgr(u)}\nonumber\\
&&\qquad\qquad\leq u\leq \nonumber\\
&&\qquad\qquad\qquad\qquad Ce^{-[\partial_r(1-\mu)(r_+,\varpi_0)-\tilde\eps]\vgr(u)}= \nonumber\\
&&\qquad\qquad\qquad\qquad Ce^{-[\partial_r(1-\mu)(r_+,\varpi_0)-\tilde\eps]\frac{v}{1+\beta}}\leq \nonumber\\
&&\qquad\qquad\qquad\qquad Ce^{-\partial_r(1-\mu)(r_+,\varpi_0)\frac{v}{1+\beta^+}},\label{P}
\end{eqnarray}
as $(u,v)\in \gam$. (The bound~\eqref{P} is actually valid in $J^-(\gam)\cap J^+(\cg_{\ckrm})$, see~\eqref{region_bis}.) Recalling~\eqref{n} and \eqref{similar2}, and using~\eqref{L} and~\eqref{O},
\begin{eqnarray*}
-\nu(u,v)&=&-\nu(u,\vgr(u))e^{\int_{\vgr(u)}^v[\kappa\partial_r(1-\mu)](u,\tilde v)\,d\tilde v}\\
&\leq&\frac{C}{u}e^{(1-\eps)\bigl[\partial_r(1-\mu)(\ckrm,\varpi_0)+\mbox{\tiny$\frac\eps{(r_--\eps_0)^2}$}\bigr]
\mbox{\tiny$\frac\beta{1+\beta}$}\, v}\\
&\leq&\frac{C}{u}e^{\partial_r(1-\mu)(\ckrm,\varpi_0)
\mbox{\tiny$\frac\beta{1+\beta^+}$}\, v}\\
&\leq&\overline Cu^{\mbox{\tiny$-\,\frac{1+\beta^-}{1+\beta^+}\frac{\partial_r(1-\mu)(\ckrm,\varpi_0)}{\partial_r(1-\mu)(r_+,\varpi_0)}\beta$}\,-1},
\end{eqnarray*}
whereas
using~\eqref{M} and~\eqref{P},
\begin{eqnarray*}
-\nu(u,v)&=&-\nu(u,\vgr(u))e^{\int_{\vgr(u)}^v[\kappa\partial_r(1-\mu)](u,\tilde v)\,d\tilde v}\\
&\geq&\frac{c}{u}e^{\partial_r(1-\mu)(r_--\eps_0,\varpi_0)
\mbox{\tiny$\frac\beta{1+\beta}$}\, v}\\
&\geq&\overline cu^{\mbox{\tiny$-\,\frac{1+\beta^+}{1+\beta^-}\frac{\partial_r(1-\mu)(r_--\eps_0,\varpi_0)}{\partial_r(1-\mu)(r_+,\varpi_0)}\beta$}\,-1}.
\end{eqnarray*}
\end{proof}

\begin{Rmk}
Since $-\partial_r(1-\mu)(r_-,\varpi_0)>\partial_r(1-\mu)(r_+,\varpi_0)$ (see\/ {\rm Appendix~A} of\/ {\rm Part~3}), we can make our choice of $\beta$ 
and other parameters ($\ckrm$, $\eps_0$, $U$)
so that 
$$\textstyle-\,\frac{1+\beta^-}{1+\beta^+}\frac{\partial_r(1-\mu)(\ckrm,\varpi_0)}{\partial_r(1-\mu)(r_+,\varpi_0)}\beta\,-1>0$$
and
$$\textstyle-\,\frac{1+\beta^+}{1+\beta^-}\frac{\partial_r(1-\mu)(r_--\eps_0,\varpi_0)}{\partial_r(1-\mu)(r_+,\varpi_0)}\beta\,-1>0.$$
Having done so, for $(u,v)$ on the curve $\gam$, we obtain
$$cu^{s_2}\leq-\nu(u,v)\leq Cu^{s_1},$$
with $0<s_1<s_2$. 
\end{Rmk}

\section{The region $J^+(\gam)$}\label{huge}

Using~\eqref{l} and~\eqref{n}, we wish to obtain upper bounds for $-\lambda$ and for $-\nu$ in the future of $\gam$ while $r$ is greater than or equal to $r_--\eps$. 
To do so, we partition this set into two regions, one where the mass is close to $\varpi_0$ and another one where
the mass is not close to $\varpi_0$. In the former case $\partial_r(1-\mu)<0$ and in the latter case $\frac{\partial_r(1-\mu)}{1-\mu}$ is bounded.
This information is used to bound the exponentials that appear in~\eqref{l} and~\eqref{n}.

Here the solution with general $\zeta_0$ departs qualitatively from the Reissner-Nordstr\"{o}m solution~\eqref{primeira}$-$\eqref{ultima}, but the radius function remains bounded away from zero, and approaches $r_-$ as $u \to 0$. This shows that the existence of a Cauchy horizon is a stable property when $\zeta_0$ is perturbed away from zero.
\begin{Lem}\label{ln-above}
Let $0<\eps_0<r_-$. There exists $0<\eps\leq\eps_0$ such that for $(u,v)\in \{r>r_--\varepsilon\}\cap J^+(\gam)$
we have
\begin{eqnarray} 
-\lambda(u,v)&\leq&Ce^{(1-\delta)\partial_r(1-\mu)(\ckrm,\varpi_0)\frac\beta{1+\beta}\,v},\label{up-l}\\
-\nu(u,v)
&\leq&Cu^{\mbox{\tiny$-\,\frac{1+\beta^-}{1+\beta^+}\frac{\partial_r(1-\mu)(\ckrm,\varpi_0)}{\partial_r(1-\mu)(r_+,\varpi_0)}\beta$}\,-1}.\label{up-n}
\end{eqnarray}
\end{Lem}
\begin{proof}
We recall that on $\gam$ the function $r$ is bounded above by $\ckrm$ and
that
$$
\eta=\eta_0+\varpi_0-\varpi.
$$
The minimum of $\eta_0$ in the interval $[r_--\eps_0,\ckrm]$ is positive, since $\eta_0(\ckrm)>0$.
If
\begin{equation}\label{varpi-m}
\varpi<\varpi_0+\min_{r\in [r_--\eps_0,\ckrm]}\eta_0(r)
\end{equation}
then clearly
\begin{equation}\label{eta-positive}
\eta>0.
\end{equation}
On the other hand, if 
\begin{equation}\label{varpi-p}
\varpi\geq\varpi_0+\min_{r\in [r_--\eps_0,\ckrm]}\eta_0(r)
\end{equation}
then, for $r\in [r_--\eps,\ckrm]$, 
$$
(1-\mu)(r,\varpi)\leq (1-\mu)(r,\varpi_0)-\,\frac{2\min_{r\in [r_--\eps_0,\ckrm]}\eta_0(r)}{\ckrm},
$$
where we used
$$
(1-\mu)(r,\varpi)=(1-\mu)(r,\varpi_0)+\frac{2(\varpi_0-\varpi)}{r}.
$$
Choosing $0<\eps\leq\eps_0$ such that
$$
\max_{r\in[r_--\eps,\ckrm]}(1-\mu)(r,\varpi_0)\leq \frac{\min_{r\in [r_--\eps_0,\ckrm]}\eta_0(r)}{\ckrm}
$$
we have
\begin{equation}\label{mu-negative}
(1-\mu)(r,\varpi)\leq -\,\frac{\min_{r\in [r_--\eps_0,\ckrm]}\eta_0(r)}{\ckrm}<0.
\end{equation}
In case~\eqref{varpi-m} we have (recall~\eqref{eta})
$$
\frac{\nu}{1-\mu}\partial_r(1-\mu)<0\qquad{\rm and}\qquad \frac{\lambda}{1-\mu}\partial_r(1-\mu)<0.
$$
In case~\eqref{varpi-p}, the absolute value of 
$$ 
-\,\frac{1}{1-\mu}\partial_r(1-\mu)
$$ 
is bounded, say by $C$. Indeed, this is a consequence of two facts: (i) the denominators $1-\mu$ and $r$ are bounded away from zero (we recall
$\eta$ also has a denominator equal to $r$);
(ii) the equality  
\begin{equation}\label{omega_bigb}
\lim_{\omega\rightarrow+\infty}-\,\frac{1}{1-\mu}\partial_r(1-\mu)=\frac 1r.
\end{equation}
We define
\begin{eqnarray*}
\Pi_v&=&
\biggl\{u\in\,]0,U]:(u,v)\in \{r>r_--\varepsilon\}\cap J^+(\gam)\ \\
&&\qquad\qquad\qquad\qquad{\rm and}\ \varpi(u,v)<\varpi_0+\min_{r\in [r_--\eps_0,\ckrm]}\eta_0(r)\biggr\},
\end{eqnarray*}
\begin{eqnarray*}
\Pi^v&=&
\biggl\{u\in\,]0,U]:(u,v)\in \{r>r_--\varepsilon\}\cap J^+(\gam)\ \\
&&\qquad\qquad\qquad\qquad{\rm and}\ \varpi(u,v)\geq\varpi_0+\min_{r\in [r_--\eps_0,\ckrm]}\eta_0(r)\biggr\},
\end{eqnarray*}
\begin{eqnarray*}
\tilde\Pi_u&=&
\biggl\{v\in\,]0,\infty[\,:(u,v)\in \{r>r_--\varepsilon\}\cap J^+(\gam)\ \\
&&\qquad\qquad\qquad\qquad{\rm and}\ \varpi(u,v)<\varpi_0+\min_{r\in [r_--\eps_0,\ckrm]}\eta_0(r)\biggr\}
\end{eqnarray*}
and
\begin{eqnarray*}
\tilde\Pi^u&=&
\biggl\{v\in\,]0,\infty[\,:(u,v)\in \{r>r_--\varepsilon\}\cap J^+(\gam)\ \\
&&\qquad\qquad\qquad\qquad{\rm and}\ \varpi(u,v)\geq\varpi_0+\min_{r\in [r_--\eps_0,\ckrm]}\eta_0(r)\biggr\}.
\end{eqnarray*}
In order to estimate $\lambda$, we observe that
\begin{eqnarray*}
&&e^{\int_{\ugam(v)}^{u}\bigl[\frac{\nu}{1-\mu}\partial_r(1-\mu)\bigr](\tilde u,v)\,d\tilde u}\\
&&\qquad\qquad= e^{\int_{\tilde u\in[\ugam(v),u]\cap\Pi_v}\bigl[\frac{\nu}{1-\mu}\partial_r(1-\mu)\bigr](\tilde u,v)\,d\tilde u}\\
&&\qquad\qquad\ \ \ \,\times e^{\int_{\tilde u\in[\ugam(v),u]\cap\Pi^v}\bigl[\frac{\nu}{1-\mu}\partial_r(1-\mu)\bigr](\tilde u,v)\,d\tilde u}\\
&&\qquad\qquad\leq 1\times e^{C\int_{\tilde u\in[\ugam(v),u]\cap\Pi^v}(-\nu)(\tilde u,v)\,d\tilde u}\\
&&\qquad\qquad\leq 1\times e^{C(\ckrm-(r_--\eps))}=:\hat C.
\end{eqnarray*}
Similarly, to estimate $\nu$ we note that
\begin{eqnarray*}
&&e^{\int_{\vgam(u)}^{v}\bigl[\frac{\lambda}{1-\mu}\partial_r(1-\mu)\bigr](u,\tilde v)\,d\tilde v}\\
&&\qquad\qquad= e^{\int_{\tilde v\in[\vgam(u),v]\cap\tilde\Pi_u}\bigl[\frac{\lambda}{1-\mu}\partial_r(1-\mu)\bigr](u,\tilde v)\,d\tilde v}\\
&&\qquad\qquad\ \ \ \,\times e^{\int_{\tilde v\in[\vgam(u),v]\cap\tilde\Pi^u}\bigl[\frac{\lambda}{1-\mu}\partial_r(1-\mu)\bigr](u,\tilde v)\,d\tilde v}\\
&&\qquad\qquad\leq 1\times e^{C\int_{\tilde v\in[\vgam(u),v]\cap\tilde\Pi^u}(-\lambda)(u,\tilde v)\,d\tilde v}\\
&&\qquad\qquad\leq 1\times e^{C(\ckrm-(r_--\eps))}=\hat C.
\end{eqnarray*} 
In conclusion, let $(u,v)\in \{r>r_--\varepsilon\}\cap J^+(\cg_{\ckrm})$. Using~\eqref{l} and~\eqref{lambda-above}, we have
\begin{eqnarray} 
-\lambda(u,v)&=&
-\lambda(\ugam(v),v)e^{\int_{\ugam(v)}^{u}\bigl[\frac{\nu}{1-\mu}\partial_r(1-\mu)\bigr](\tilde u,v)\,d\tilde u}\label{lambda_above}\\
&\leq&\hat C
\tilde Ce^{(1-\delta)\partial_r(1-\mu)(\ckrm,\varpi_0)\frac\beta{1+\beta}\,v}.\nonumber
\end{eqnarray}
Similarly, using~\eqref{n} and~\eqref{nu-above}, we have
\begin{eqnarray*}
-\nu(u,v)&=&
-\nu(u,\vgam(u))e^{\int_{\vgam(u)}^{v}\bigl[\frac{\lambda}{1-\mu}\partial_r(1-\mu)\bigr](u,\tilde v)\,d\tilde v}\\
&\leq&\hat C
\overline Cu^{\mbox{\tiny$-\,\frac{1+\beta^-}{1+\beta^+}\frac{\partial_r(1-\mu)(\ckrm,\varpi_0)}{\partial_r(1-\mu)(r_+,\varpi_0)}\beta$}\,-1}.
\end{eqnarray*}
\end{proof}

\begin{Lem}\label{infty}
 Let $\delta>0$. There exists $\tilde U_\delta$ such that for $(u,v)\in J^+(\gam)$ with $u<\tilde U_\delta$, we have
 $$
 r(u,v)>r_--\delta.
 $$
 \end{Lem}
\begin{proof}
We denote by $\overline\eps$ the value of $\eps$ that is provided in Lemma~\ref{ln-above}.
Let $\delta>0$. Without loss of generality, we assume that $\delta$ is less than or equal to $\overline\eps$.
Choose the value of $\eps$ in Corollary~\ref{cor-8} equal to $\delta$. This determines an $U_\eps$ as in the statement of that corollary.
Let $(u,v)\in J^+(\gam)$ with $u\leq U_\eps$.
Then
 $$
 r(u,\vgam(u))\geq r_--\frac\delta 2\qquad{\rm and}\qquad r(\ugam(v),v)\geq r_--\frac\delta 2
 $$
because $\ugam(v)\leq u$.
Here
$$
u\mapsto (u,\vgam(u))\quad {\rm and}\quad v\mapsto (\ugam(v),v)
$$
are parameterizations of the spacelike curve $\gam$.
Integrating~\eqref{up-n}, we obtain
\begin{equation}\label{ok}
-\int_{\ugam(v)}^u\frac{\partial r}{\partial u}(s,v)\,ds\leq
\int_{\ugam(v)}^uCs^{p-1}\,ds,
\end{equation}
for a positive $p$. This estimate is valid for $(u,v)\in \{r>r_--\overline\eps\}\cap J^+(\gam)$. It yields
\begin{eqnarray}
r(u,v)&\geq& r(\ugam(v),v)-\frac Cp(u^p-(\ugam(v))^p)\nonumber\\
&\geq&r_--\frac\delta 2-\frac Cpu^p\ >\ r_--\delta,\label{r-below}
\end{eqnarray}
provided $u<\min\Bigl\{U_\eps,\sqrt[p]{\frac{\delta p}{2C}}\Bigr\}=:\tilde U_\delta$. 
Since $\delta$ is less than or equal to $\overline\eps$ and $\gam\subset \{r>r_--\overline\eps\}$, if $(u,v)\in J^+(\gam)$ and $u<\tilde U_\delta$, then
$(u,v)\in \{r>r_--\overline\eps\}$ and 
the estimate~\eqref{ok} does indeed apply.

Alternatively, we can obtain~\eqref{r-below} integrating~\eqref{up-l}:
$$
-\int_{\vgam(u)}^v\frac{\partial r}{\partial v}(u,s)\,ds\leq
\int_{\vgam(u)}^vCe^{-qs}\,ds,
$$
for a positive $q$. This yields
\begin{eqnarray*}
r(u,v)&\geq& r(u,\vgam(u))-\frac Cq\Bigl(e^{-q\vgam(u)}-e^{-qv}\Bigr)\\
&\geq&r_--\frac\delta 2-\frac Cqe^{-q\vgam(u)}\\ &\geq&\ r_--\frac\delta 2-{\tilde C}u^{\tilde q},
\end{eqnarray*}
for a positive $\tilde q$, according to~\eqref{O}. For $u<\min\Bigl\{U_\eps,\sqrt[\tilde q]{\frac{\delta}{2\tilde C}}\Bigr\}$
we obtain, once more,
$$
r(u,v)>r_--\delta.
$$
\end{proof}

\begin{Cor}
If $\delta<r_-$ then ${\cal P}$ contains $[0,\tilde U_\delta]\times[0,\infty[$. Moreover, estimates~\eqref{up-l} and~\eqref{up-n} hold on~$J^+(\gam)$.
\end{Cor}
Due to the monotonicity of $r(u,\,\cdot\,)$ for each fixed $u$, we may define
$$
r(u,\infty)=\lim_{v\to\infty}r(u,v).
$$
As $r(u_2,v)<r(u_1,v)$ for $u_2>u_1$, we have that $r(\,\cdot\,,\infty)$ is nonincreasing.
\begin{Cor}\label{Gana}
We have
\begin{equation}\label{r-menos} 
\lim_{u\searrow 0}r(u,\infty)=r_-.
\end{equation}
\end{Cor}

The previous two corollaries prove Theorem~\ref{r-stability}.
The argument in~\cite[Section~11]{Dafermos2}, shows that, as in the case when $\Lambda=0$, the spacetime is then extendible across the Cauchy horizon with $C^0$ metric.

\section{Two effects of any nonzero field}

This section contains two results concerning the structure of the solutions with general $\zeta_0$.
Theorem~\ref{rmenos} asserts that only in the case of the Reissner-Nordstr\"{o}m solution does the curve $\cg_{r_-}$ coincide with
the Cauchy horizon: if the field $\zeta_0$ is not identically zero, then the curve $\cg_{r_-}$ is contained in ${\cal P}$.

Lemma~\ref{l-kappa} states that, in contrast with what happens with the Reissner-Nordstr\"{o}m solution,
and perhaps unexpectedly,
the presence of a nonzero field immediately causes the integral $\int_0^\infty \kappa(u,v)\,dv$ to be finite for any $u>0$.
This implies that the affine parameter of any outgoing null geodesic inside the event horizon is finite at the Cauchy horizon.

For each $u>0$, we define
\begin{eqnarray*}
\varpi(u,\infty)&=&\lim_{v\nearrow+\infty}\varpi(u,v).
\end{eqnarray*}
This limit exists, and $u\mapsto\varpi(u,\infty)$ is an increasing function.

\begin{Thm}\label{rmenos} 
 Suppose that there exists a positive sequence $(u_n)$ converging to $0$ such that $\zeta_0(u_n)\neq 0$.
 Then $r(u,\infty)<r_-$ for all $u \in \left] 0, U \right]$. 
\end{Thm}

\begin{proof}
 The proof is by contradiction. Assume that $r(\bar{u},\infty) = r_-$ for some $\bar{u} \in \left]0,U\right]$.
Then $r(u,\infty) = r_-$ for all $u \in \left]0, \bar{u}\right]$. Let $0<\delta<u\leq \bar{u}$. Clearly,
$$
r(u,v)=r(\delta,v)+\int_\delta^u\nu(s,v)\,ds.
$$
Fatou's Lemma implies that $$\liminf_{v\to\infty}\int_\delta^u-\nu(s,v)\,ds\geq \int_\delta^u\liminf_{v\to\infty}-\nu(s,v)\,ds.$$
So,
\begin{eqnarray}
r_-=\lim_{v\to\infty}r(u,v)&=&\lim_{v\to\infty}r(\delta,v)-\lim_{v\to\infty}\int_\delta^u-\nu(s,v)\,ds\nonumber\\
&=&r_--\liminf_{v\to\infty}\int_\delta^u-\nu(s,v)\,ds\nonumber\\
&\leq&r_--\int_\delta^u\liminf_{v\to\infty}-\nu(s,v)\,ds.\label{fatou}
\end{eqnarray}
Since $\delta$ is arbitrary, this inequality implies that $\liminf_{v\to\infty}-\nu(u,v)$ is equal to zero for almost all $u \in \left]0, \bar{u}\right]$. 
However, we will now show that, under the hypothesis on $\zeta_0$, $\liminf_{v\to\infty}-\nu(u,v)$ cannot be zero for any positive $u$ if $r(u,\infty)\equiv r_-$.
 
First, assume that $\varpi(u,\infty)=\infty$ for a certain $u$.
Then, using~\eqref{omega_bigb},
$$\lim_{v\rightarrow\infty}\frac{\partial_r(1-\mu)}{1-\mu}(u,v)=-\frac{1}{r_-}<0.$$
We may choose $V=V(u)>0$ such that $\frac{\partial_r(1-\mu)}{1-\mu}(u,v)<0$ for $v>V$. Integrating~\eqref{nu_v}, for $v>V$,
 \begin{eqnarray*}
  -\nu(u,v) 
  &=&
  -\nu(u,V)e^{\int_V^v\left[\frac{\partial_r(1-\mu)}{1-\mu}\lambda\right](u,\tilde v)\,d\tilde v}
  \\
  &\geq&
  -\nu(u,V)>0.
 \end{eqnarray*}
Thus, for such a $u$, it is impossible for $\liminf_{v\to\infty}-\nu(u,v)$ to be equal to zero.

 Now assume $\varpi(u,\infty)<\infty$. The hypothesis on $\zeta_0$ and~\eqref{omega_u} imply that
$\varpi(u,0)>\varpi_0$ for each $u>0$, and so 
  $\varpi(u,\infty)>\varpi_0$ for each $u>0$.  Then, 
 $$(1-\mu)(u,\infty)=(1-\mu)(r_-,\varpi(u,\infty))< (1-\mu)(r_-,\varpi_0)=0.$$
We may choose $V=V(u)>0$ such that $-(1-\mu)(u,v)\geq C(u)>0$ for $v>V$.
Hence, integrating~\eqref{omega_v}, for $v>V$,
\begin{eqnarray*}
\varpi(u,v)&=&\varpi(u,V)+\frac 12\int_V^{v}\Bigl[-(1-\mu)\frac{\theta^2}{-\lambda}\Bigr](u,v)\,dv\\
&\geq&\varpi(u,V)+\frac {C(u)}2\int_V^{v}\Bigl[\frac{\theta^2}{-\lambda}\Bigr](u,v)\,dv
\end{eqnarray*}
Since $\varpi(u,\infty)<\infty$, letting $v$ tend to $+\infty$, we conclude
$$\int_V^{\infty}\Bigl[\frac{\theta^2}{-\lambda}\Bigr](u,v)\,dv<\infty.$$
 Finally, integrating~\eqref{ray_v_bis} starting from $V$, we see that $\frac{\nu(u,\infty)}{(1-\mu)(u,\infty)}>0$. Since $(1-\mu)(u,\infty)<0$,
once again we conclude that $\liminf_{v\to\infty}-\nu(u,v)=-\nu(u,\infty)>0$.
 \end{proof}

\begin{Lem}\label{l-kappa} 
 Suppose that there exists a positive sequence $(u_n)$ converging to $0$ such that $\zeta_0(u_n)\neq 0$.
Then 
\begin{equation}\label{integral-k} 
\int_{0}^\infty\kappa(u,v)\,dv<\infty\ {\rm for\ all}\ u>0.
\end{equation}
\end{Lem}
\begin{proof}
We claim that for some decreasing sequence $(u_n)$ converging to 0, $$(1-\mu)(u_n,\infty)<0.$$
To prove our claim, we consider three cases.

{\em Case 1}. If $\varpi(u,\infty)=\infty$ for each $u>0$ then $(1-\mu)(u,\infty)=-\infty$.

{\em Case 2}. If $\lim_{u\searrow 0}\varpi(u,\infty)>\varpi_0$ then, using Corollary~\ref{Gana},
$$
\lim_{u\searrow 0}(1-\mu)(u,\infty)=(1-\mu)(r_-,\lim_{u\searrow 0}\varpi(u,\infty))<(1-\mu)(r_-,\varpi_0)=0.
$$

{\em Case 3}. 
Suppose that
 $\lim_{u\searrow 0}\varpi(u,\infty)=\varpi_0$.
For sufficiently small $u$ and $(u,v)\in J^+(\cg_{\ckrm})$, we have 
$$ 
\eta(u,v)\geq 0
$$ 
(see~\eqref{eta-positive}).
So, we may define
$\nu(u,\infty)=\lim_{v\nearrow+\infty}\nu(u,v)$. By Lebesgue's Monotone Convergence Theorem, we have
\begin{equation}\label{lebesgue-2} 
r(u,\infty)=r(\delta,\infty)+\int_\delta^u\nu(s,\infty)\,ds.
\end{equation}
Note that different convergence theorems have to be used in~\eqref{fatou} and~\eqref{lebesgue-2}.
If $\nu(u,\infty)$ were zero almost everywhere, then $r(u,\infty)$ would be a constant. If the constant were $r_-$ we would be contradicting Theorem~\ref{rmenos}.
If the constant were smaller than $r_-$ we would be contradicting Lemma~\ref{infty}.
We conclude there must exist a sequence $u_n\searrow 0$ such that
$\nu(u_n,\infty)<0$. 
Integrating~\eqref{ray_v_bis}, we get
$$ 
\frac{\nu(u,\infty)}{(1-\mu)(u,\infty)}\leq\frac{\nu(u,0)}{(1-\mu)(u,0)}<\infty.
$$ 
Therefore, $(1-\mu)(u_n,\infty)<0$. This proves our claim.

For any fixed index $n$, there exists a $v_n$ such that
$$(1-\mu)(u_n,v)<\frac 12(1-\mu)(u_n,\infty)=:-\,\frac 1{c_n},$$ for $v\geq v_n$. 
It follows that
$$
\kappa(u_n,v)\leq c_n(-\lambda(u_n,v)),\ {\rm for}\ v\geq v_n.
$$
Using the estimate~\eqref{up-l} for $-\lambda$, we have 
$$
\int_{v_n}^\infty\kappa(u_n,v)\,dv<\infty.
$$
Hence
$
\int_{0}^\infty\kappa(u_n,v)\,dv<\infty
$.
Recalling that $u\mapsto\kappa(u,v)$ is nonincreasing, we get~\eqref{integral-k}.
\end{proof}

\begin{Cor}\label{affine}
Let $u>0$. Consider an outgoing null geodesic $t\mapsto(u,v(t))$ for $({\cal M},g)$, with $g$ given by
$$
g=-\Omega^2(u,v)\,dudv+r^2(u,v)\,\sigma_{{\mathbb S}^2}.
$$
Then $v^{-1}(\infty)<\infty$, i.e.\
the affine parameter is finite at the Cauchy horizon.
\end{Cor}
\begin{proof}
The function $v(\,\cdot\,)$ satisfies
\begin{equation}\label{geodesic}
\ddot{v}+\Gamma^v_{vv}(u,v)\,\dot{v}^2=0,
\end{equation}
where the Christoffel symbol $\Gamma^v_{vv}$ is given by
$$\Gamma^v_{vv}=\partial_v\ln\Omega^2.$$
So, we may rewrite~\eqref{geodesic} as 
$$
\frac{\ddot{v}}{\dot{v}}=-\partial_t(\ln\Omega^2)(u,v).
$$
We integrate both sides of this equation to obtain
$$
\ln \dot v+\ln c=-\ln\Omega^2(u,v),
$$
with $c>0$, or
$$
\frac{dt}{dv}=c\,\Omega^2(u,v).
$$
Integrating both sides of the previous equation once again, the affine parameter $t$ is given by
$$
t=v^{-1}(0)+c\int_0^v\Omega^2(u,\bar v)\,d\bar v=v^{-1}(0)-4c\int_0^v(\nu\kappa)(u,\bar v)\,d\bar v.
$$

If $\zeta_0$ vanishes in a neighborhood of the origin, the solution corresponds to the Reissner-Nordstr\"{o}m solution. The function
$\kappa$ is identically 1 and, using~\eqref{ray_v_bis}, $\frac{\nu}{1-\mu}=C(u)$, with $C(u)$ a positive function of $u$.
Thus, $\nu=C(u)(1-\mu)=C(u)\lambda$ and 
$$
\int_0^\infty\Omega^2(u,\bar v)\,d\bar v=-4cC(u)\int_0^\infty\lambda(u,\bar v)\,d\bar v=4cC(u)(r_+-u-r_-)<\infty.
$$

On the other hand, suppose that there exists a positive sequence $(u_n)$ converging to $0$ such that $\zeta_0(u_n)\neq 0$.
Then, since $\nu$ is continuous, it satisfies the bound~\eqref{up-n} for large $v$, and~\eqref{integral-k} holds. So we also
have
$$
\int_0^\infty \Omega^2(u,\bar v)\,d\bar v<\infty.
$$
\end{proof}

\end{document}